\pdfoutput=1
\ifx\synctex\undefined\else\synctex=1\fi
\documentclass[copyright,creativecommons]{eptcs}
\providecommand{\event}{GandALF 2016} %

\usepackage[T1]{fontenc}
\usepackage[utf8]{inputenc}
\usepackage{microtype}

\usepackage{mathtools}
\usepackage{xfrac}
\usepackage{paralist}
\usepackage{bm}
\usepackage{amssymb}
\usepackage{amsthm}
\usepackage{url}

\usepackage[vlined]{algorithm2e}

\usepackage{graphicx}

\usepackage{tikz-dependency}

\newtheorem{theorem}{Theorem}
\newtheorem{proposition}{Proposition}
\newtheorem{example}{Example}[section]
\newtheorem{lemma}{Lemma}
\newtheorem{definition}{Definition}

\newcommand{\N}{\mathbb{N}}

\newcommand{\pr}{R}

\newcommand{\osc}{\mathit{osc}}
\renewcommand{\dim}{\mathit{dim}}
\newcommand{\ID}{\mathit{I}}
\newcommand{\rank}{\mathit{rank}}
\newcommand{\stcomp}[1]{\overline{#1}}
\newcommand{\ann}{\mathit{annot}}

\newcommand{\state}{\mathit{state}}
\newcommand{\tape}{\mathit{tape}}
\newcommand{\stack}{\mathit{stack}}
\newcommand{\yield}{\mathbb{Y}}

\makeatletter

\newcommand{\Rmnum}[1]{\expandafter\@slowromancap\romannumeral #1@}
\makeatother

\def\len#1{{\vert{#1}\vert}}

\title{Bounded-oscillation Pushdown Automata}
\def\titlerunning{Bounded-oscillation Pushdown Automata}

\author{Pierre Ganty
\institute{IMDEA Software Institute\\ Madrid, Spain}
\email{pierre.ganty@imdea.org}
\and
Damir Valput
\institute{IMDEA Software Institute\\ Madrid, Spain}
\email{damir.valput@imdea.org}
}

\def\authorrunning{P. Ganty and D. Valput}

\begin{document}
\maketitle

\begin{abstract}
 We present an underapproximation for context-free languages by filtering out
 runs of the underlying pushdown automaton depending on how the stack height
 evolves over time. In particular, we assign to each run a number quantifying
 the oscillating behavior of the stack along the run. We study languages
 accepted by pushdown automata restricted to \(k\)-oscillating runs.  We relate
 oscillation on pushdown automata with a counterpart restriction on
 context-free grammars.  We also provide a way to filter all but the
 \(k\)-oscillating runs from a given PDA by annotating stack symbols with
 information about the oscillation. Finally, we study closure properties of the
 defined class of languages and the complexity of the \(k\)-emptiness problem
 asking, given a pushdown automaton \(P\) and \(k\geq 0\), whether \(P\) has a
 \(k\)-oscillating run.  We show that, when \(k\) is not part of the input, the
 \(k\)-emptiness problem is NLOGSPACE-complete.
\end{abstract}

\section{Introduction}
Since the inception of context-free languages (CFLs for short), researchers have studied their
properties including how to define “well-behaved” subclasses.  Typically,
subclasses are obtained by posing restrictions excluding some behaviors of
the underlying formalism (context-free grammar or pushdown automaton).
For instance, visibly pushdown automata \cite{DBLP:journals/jacm/AlurM09}
require input symbols to dictate push or pop operations on the stack. Another
restriction is bounding the number of turns~\cite{Ginsburg_1966}---switches
from non-decreasing to non-increasing modes---of the stack over time. 

In all those cases, restrictions are trying to achieve one or more of the following objectives:
\begin{inparaenum}[\upshape(\itshape i\upshape)]
	\item capture a large subset of context-free languages; 
	\item define a subclass with good closure properties (e.g. closure
		to boolean operations, to homomorphism or their inverse, \ldots)
	\item obtain more efficient algorithms (e.g. for parsing);
	\item obtain new decidability results (e.g. language equivalence). 
\end{inparaenum}

In this paper, we define a new restriction that generalizes finite-turn
pushdown automata \cite{Ginsburg_1966}.  Our restriction is based on the
non-trivial yet natural notion of \emph{oscillation}, a measure of how variable
is the stack height over time. To get a glimpse of how our restriction
generalizes finite-turn consider the language \( L = \{ (a^n b^n)^* \mid n\geq
0\}\). A pushdown automaton (PDA for short) deciding \(L\) has to keep
track---using its stack---of the number of symbols ‘\(a\)’: reading an ‘\(a\)’ results
in a push while reading a ‘\(b\)’ results in a pop.  By limiting the number of
turns it thus seems difficult, if at all possible, to capture \(L\): e.g. one
turn allows to capture precisely \(L\cap a^* b^* \), two turns \(L\cap a^* b^*
a^* b^* \), \ldots However no finite number of turns captures \(L\). On the
other hand, restricting the runs of that PDA to those which have an oscillation
of at most \(1\) is enough to capture \(L\) entirely.

The oscillation of PDA runs is defined using a hierarchy of so-called
harmonics.  Harmonics are prototypical sequence of stack moves: order \(1\)
harmonic is \textit{push pop push pop}, order \(2\) harmonic is \textit{push
<order 1 harmonic>  pop push <order 1 harmonic> pop} etc. 
Hence, we say that a PDA run \(r\) is \(k\)-oscillating if the harmonic of
order \(k\) is the greatest harmonic that can be “found” in \(r\).

Equipped with the restriction based on oscillation, we evaluate the
aforementioned objectives.  In particular,
\begin{compactitem}
	\item we study closure properties to boolean operations of the bounded-oscillation languages, we show they are not determinizable and that the problem whether a given context-free language is bounded-oscillation is undecidable. 
	\item we study the \(k\)-emptiness problem which asks, given a PDA
		and a number \(k\),  whether there exists a \(k\)-oscillating run. We show that,
		when \(k\) is not part of the input, the \(k\)-emptiness problem is
		NLOGSPACE-complete. A slight adaptation of the given algorithm solves the
		\(k\)-membership problem: given a PDA \(P\), \(k\), and a word \(w\), does there exist a
		\(k\)-oscillating run of \(P\) accepting input \(w\).
	\item we relate oscillation on PDA with a counterpart
		restriction on context-free grammars. 
		This allows to reformulate some results and their proofs using PDAs instead of context-free grammars. Such reformulations are out of the scope of the paper but let us evoke some possibilities. 
		For example, decidability
		and complexity results in computing procedure summaries for a class of procedural programs~\cite{conf/tacas/GantyIK13}.  
		Also, the decidability of the reachability
		problem for a subclass of Petri nets extended with a stack~\cite{conf/fsttcs/AtigG11}. The previously cited works (also
		\cite{journals/toplas/EsparzaGP14}) sometimes include an unnecessary step translating
		from PDA to CFG and back. Thanks to the relation we prove, translations
		back and forth can be omitted thus obtaining more direct proofs.
\end{compactitem}

As a collateral contribution, let us mention that our proofs propose a novel
framework allowing to reason uniformly about parse trees and PDA runs through
the use of well-parenthesized words. Doing so, we obtain objects which are
simple, and intuitively easy to understand. Incidentally, proofs turn out to be
elegantly simple.  

Finally, we provide a syntactic characterization of bounded-oscillation pushdown
automata in the following sense: given a PDA \(P\) and a number \(k\) we show
that by modifying its stack alphabet and PDA actions, but by keeping unchanged its
input alphabet, we obtain another PDA for the residual language of \(P\) where
only the \(k\)-oscillating runs of \(P\) have been kept.  Because the previous
construction preserves the nature of PDA actions in the sense that a push
remains a push and a pop remains a pop, applying it to visibly pushdown
automata results into visibly pushdown automata with only
\(k\)-oscillating runs.

Missing proofs are given in the appendix in the full version of this paper.

\subparagraph*{Related work.}
Nowotka and Srba \cite{conf/mfcs/NowotkaS07} considered a subclass of PDA they
call height-deterministic pushdown automata. Unlike their class, our class
imposes restrictions on the evolution of the stack over time regardless of the
input.

For context-free language specified by grammars, Esparza et
al.~\cite{conf/dlt/EsparzaKL07,journals/ipl/EsparzaGKL11} relate two measures:
the so-called dimension defined on parse trees and the index defined on
derivations. Luttenberger and Schlund~\cite{journals/iandc/LuttenbergerS16}
made a step further and proved the dimension and path-width of parse trees are
in linear relationship. Our work instates the notion of oscillation defined on
PDA runs and establishes a linear relationship with the dimension of
parse trees. In the process, we address the challenge of connecting notions
formulated for equivalent yet different formalisms.

Wechsung \cite{conf/fct/Wechsung79} studies PDA runs by representing them in
the 2-dimensional plane. Through a graphical
notion of derivative applied on the representation of PDA runs, he formulates
a notion of oscillation.  Although Wechsung provides critical insights on
oscillation, his definitions are ambiguous and lack proper formalization. We go
further by proposing a clean, formal, and language based definition of
oscillation.

\section{Preliminaries}
An \emph{alphabet} \(\Sigma\) is a nonempty finite set of \emph{symbols}. A
word \(w\) is a finite sequence of symbols of \(\Sigma\), i.e \(w\in\Sigma^*\).
We denote by \(\len{w}\) the length of \(w\). Further define \( (w)_i \) as the
\(i\)-th symbol of \(w\) if \(1 \leq i \leq \len{w}\) and \(\varepsilon\)
otherwise. Hence, \(w = (w)_1 \ldots (w)_{\len{w}}\). A language is a set of
words.

A \emph{pushdown automaton} (PDA) is a tuple \( (Q, \Sigma, \Gamma, \delta, q_0, \gamma_0)\) where:
\begin{compactitem}
\item $Q$ is finite set of \emph{states} including \(q_0\), the \emph{start state};
\item $\Sigma$ is an alphabet called \emph{input alphabet};
\item $\Gamma$ is finite set of \emph{stack symbols} (\emph{pushdown alphabet}) including \(\gamma_0\), the \emph{start stack symbol};
\item $\delta$ is a finite subset of \( Q \times (\Sigma\cup\{\varepsilon\}) \times \Gamma \times Q \times \Gamma^*\). We individually refer to each element of \(\delta\) as an \emph{action} and use the notation
	\( (q,b,\gamma) \hookrightarrow (p,\xi) \) for an action \( (q,b,\gamma,p,\xi) \in \delta \).
\end{compactitem} 

An \emph{instantaneous descriptor} (ID) of a PDA \(P\) is a triple $(q, w,
\xi)$ where $q$ is the state of the PDA, $w$ is the input word left to
read, and $\xi$ is the stack content.  Given an \emph{input word}
\(w\), we define the \emph{initial ID} of \(P\) to be \( (q_0,w, \gamma_0) \) and denote it
\(\ID_s(w)\). Given an ID \(\ID=(q,w,\xi)\), define \( \state(\ID)\), \(\tape(\ID)\)
and \(\stack(\ID)\) to be \(q\), \(w\) and \(\xi\), respectively.

Given an action \( (q,b,\gamma) \hookrightarrow (p,\xi') \) and an ID \((q, bw, \gamma
\xi)\) of \(P\) define a \emph{move} to be $(q, bw, \gamma \xi) \vdash_P (p,
w, \xi' \xi)$.  We often omit the subscript \(P\) when it is clear from
the context.
A \emph{move sequence} of \(P\) is finite sequence $\ID_0, \ID_1, \ldots, \ID_m$
where \(m\geq 0\) of IDs such that \(\ID_i \vdash_P \ID_{i+1}\) for all \(i\).
We respectively call \(\ID_0\) and \(\ID_m\) the \emph{first} and \emph{last}
ID of the move sequence and write \(\ID_0 \vdash^*_P \ID_m \) to denote a
move sequence from \(\ID_0\) to \(\ID_m\) whose intermediate IDs are not important.
A \emph{quasi-run} $r$ of \(P\) is a move sequence \(\ID \vdash^*_P \ID'\) such
that \(\stack(\ID)\in\Gamma\) and \(\stack(\ID')=\varepsilon\).
A \emph{run} \(r\) of \(P\) on input \(w\in\Sigma^*\) is a quasi-run  \(\ID
\vdash^*_P \ID'\) where \(\ID=\ID_s(w)\) and \(\tape(\ID')= \varepsilon\).
Intuitively, a run is a quasi-run that starts from ID \(\ID_s(w)\) and reads
all of \(w\).
We say that a word $w \in \Sigma^{*}$ is accepted by \(P\) if there exists a
run on input \(w\). The \emph{language} of \(P\), denoted \(L(P)\), is the set
of words for which \(P\) has a run. Formally, \( L(P) = \{w \in \Sigma^{*} \mid
\ID_s(w) \vdash^{*}_{P} \ID \text{ and } \tape(\ID)=\stack(\ID)=\varepsilon
\}\).

A \emph{context-free grammar} (CFG or grammar for short) is a tuple $G = (V,
\Sigma, S, \pr)$ where \(V\) is a finite set of \emph{variables} (or
\emph{non-terminals}) including the \emph{start variable} \(S\); \(\Sigma\) is
an alphabet (or set of \emph{terminals}), $\pr \subseteq V \times (\Sigma \cup
V)^{*}$ is a finite set of \emph{rules}. We often write \(X \rightarrow w\) for
a rule \( (X,w) \in \pr\).  We define a \emph{step} as the binary relation
$\Rightarrow_G$ on $(V \cup \Sigma)^{*}$ given by \(u \Rightarrow_G v\) if there
exists a rule \(X \rightarrow w\) of \(G\), \( (u)_i = X \) and \( v =
(u)_{1}\ldots (u)_{i-1} w (u)_{i+1} \ldots (u)_{\len{u}}\). We call \(i\) as the \emph{position selected} by the step.
Define \(u \Rightarrow^{*}_G v\) if there exists a step sequence \(u_0 \Rightarrow_G u_1 \Rightarrow_G \ldots \Rightarrow_G u_n\) 
such that \(u_0 = u\) and \(u_n = v\). A step sequence \(u \Rightarrow_G^* w\) is called a \emph{derivation} whenever
\(u = S\) and \(w\in \Sigma^*\). A step sequence \(u_0 \Rightarrow_G u_1 \Rightarrow_G \ldots \Rightarrow_G u_n\) is said to be \emph{leftmost}
if for each step \(u_i \Rightarrow_G u_{i+1}\), the position \(p_i\) selected is such that \( (u_i)_j \in V \) for no \(j < p_i\).
Define \(L(G) = \{ w \in \Sigma^* \mid S \Rightarrow^{*}_{G} w \} \) and call it the language generated by \(G\).

\noindent
Given a grammar \( (V,\Sigma,S,\pr) \) and \(Z\in V \cup \Sigma \cup \{\varepsilon\}\), define a quasi parse tree (or \emph{quasi-tree} for short), denoted \(t_Z\), to be a tree satisfying:
\begin{compactitem}
\item \(Z\) labels the root of \(t\); and
\item Each interior node is labelled by a variable; and
\item Each leaf is labelled by either a terminal $b \in \Sigma$ or $\varepsilon$. If the leaf is labelled $\varepsilon$, then it must be the only child of its parent (if any); and
\item If an interior node is labelled by $X$, and its \(k\) children are labelled \(X_1\) to \(X_k\), in that order, then $X \to X_1 X_2 \ldots X_k$ is a rule in \(R\).
\end{compactitem}
Next we define a \emph{parse tree} to be a quasi-tree with root \(S\)---the
start variable of \(G\).  Observe that all parse trees have at least two nodes
while quasi-trees have at least one. Also when the root of a quasi-tree is
labelled with \(a\in\Sigma\) or \(\varepsilon\) then it contains no other
nodes.
Given a quasi-tree \(t\) define its yield, denoted \(\yield(t)\), to be the word over
\(\Sigma\) obtained by concatenating the labels of the leaves of \(t\) from
left to right.

To each node \(n\) in a tree $t$ we assign a \emph{dimension} \(\dim(n)\) as follows:
\begin{compactitem}
\item If \(n\) is a leaf, then \(\dim(n) = 0\).
\item If \(n\) has children $n_1, n_2, \ldots, n_k$ with \(k \geq 1\) then
\[
\dim(n) = 
\begin{cases}
	\max_{i\in \{1,\ldots,k\}} \dim(n_i)  & \text{if there is a unique maximum}\\
  \max_{i\in \{1,\ldots,k\}} \dim(n_i)+1& \text{otherwise}
\end{cases}
\]
\end{compactitem}
We define the \emph{dimension of a tree \(t\)} with root $n$, denoted \(\dim(t)\), as \(\dim(n)\).

\begin{example}\label{ex:dyck}

	\hspace*{0pt}\\
	\noindent
	\begin{minipage}[b]{.6\textwidth}
\begin{tikzpicture}[thick,level/.style={sibling distance=20mm/#1},level distance=6mm]
    \node {\(S\)}
		    child {
					node {\(\bar{a}\)}
				} child {
           node {\(S\)}
           child { node {\(\varepsilon\)} }
        } child {
						node {\(a\)}
				} child {
						node {\(S\)}
						child {
							node {\(\bar{a}\)}
						} child {
							 node {\(S\)}
							 child {
									 node {\(\bar{a}\)}
							 } child {
									 node {\(S\)}
									 child { node {\(\varepsilon\)} }
							 } child {
							     node {\(a\)}
							 } child {
									 node {\(S\)}
							     child { node {\(\varepsilon\)} }
							 } 
						} child {
							node {\(a\)}
						} child {
                node {\(S\)}
							  child { node {\(\varepsilon\)} }
						} 
       } 
    ;
\end{tikzpicture}
	\end{minipage}%
	\begin{minipage}[b]{.4\textwidth}
	Let \( G_D = (\{S\}, \{\bar{a},a\}, S,\{ S \rightarrow \bar{a}\, S\, a\, S, S \rightarrow \varepsilon\}) \).
	We denote \(L(G_D)\) by \(L_D\), the \emph{Dyck language} over \( (\bar{a},a) \). 
	A parse tree \(t\) for the word \( \bar{a}\, a\, \bar{a}\, \bar{a}\, a\, a\) such that \(\dim(t)=1\) is given left.
	\end{minipage}%
\end{example}

\section{Oscillation For Trees: a Dyck Word Based Approach}
In this section, we match trees with Dyck words and define a measure based on a partial
ordering on Dyck words and special Dyck words we call harmonics. We start by
recalling that \(G_D\) is unambiguous.

\begin{proposition}\label{prop:unambiguousDyck}
The CFG \(G_D= (\{S\}, \{\bar{a},a\}, S,\{ S \rightarrow \bar{a}\, S\, a\, S, S \rightarrow \varepsilon\})\) is unambiguous.
\end{proposition}
Let \(w \in L_D\) and let \(t\) be its unique corresponding parse tree.
Unambiguity of \(G_D\) enables us to elegantly define matching pairs inside a word \(w \in L_D\).
Two positions \(i < j\) form a \emph{matching pair} \( (i,j)\) if \( (w)_i = \bar{a}
\), \( (w)_j = a \) and the two leaves corresponding to \( (w)_i \) and \( (w)_j \)
in \(t\) have the same parent.

\begin{example}
	Consider the parse tree \(t\) of Example~\ref{ex:dyck} and the word \( w=\bar{a}\, a\, \bar{a}\, \bar{a}\, a\, a\) it defines. The matching pairs
	of \(w\) are given by \(\{ (1,2), (3,6), (4,5) \}\). We prefer to use the more intuitive representation where the endpoints of the arrows
	are the matching pairs:
		\begin{dependency}[anchor=base, baseline=0, theme=simple, arc edge , arc angle = 15]
			\begin{deptext}[column sep=0.2cm]
				\(\bar{a}\) \& \(a\) \& \(\bar{a}\) \& \(\bar{a}\) \& \(a\) \& \(a\) \\
			\end{deptext}
			\depedge{1}{2}{}
			\depedge{3}{6}{}
			\depedge{4}{5}{}
		\end{dependency}
		\qed
\end{example}
Thus, we can determine following properties of matching pairs:
\begin{compactitem}
\item Arrows can only go forward: each matching pair \( (i,j) \) is such that \(i < j\). 
\item For each word \(w\in L_D\) and each position \(p\) in \(w\), if \( (w)_p = \bar{a} \) then there is exactly one arrow leaving from \(p\); else ( \((w)_p = a\) ) there is exactly one arrow ending in \(p\).
\item Arrows cannot cross: no two matching pairs \((i_1, j_1)\) and \((i_2, j_2)\) are such that
	\( i_1 < i_2 < j_1 < j_2\). Graphically, the following is forbidden:
		\begin{dependency}[anchor=base, baseline=0, theme=simple, arc edge , arc angle = 10]
			\begin{deptext}[column sep=0.2cm]
				\(\bar{a}\) \& \ldots \& \(\bar{a}\) \& \ldots \& \(a\) \& \ldots \& \(a\) \\
			\end{deptext}
			\depedge{1}{5}{}
			\depedge{3}{7}{}
		\end{dependency}
\end{compactitem}
Given two words \(w_a\) and \(w_b\) of \(L_D\), define the ordering \(w_a \preceq
w_b\) to hold whenever $w_a$ results from $w_b$ by deleting \(0\) or more
matching pairs.

\begin{example}
For $w_a = \bar{a}\; \bar{a} \;a \;a$ and $w_b =
\bar{a}\;\bar{a}\;\bar{a}\;a\;\bar{a}\;a\;\bar{a}\;a \;a \;a $; \(w_a,w_b\in L_D\) the ordering $w_a \preceq w_b$ holds
since $w_a$ results from deleting the three matching pairs
in $w_b$ depicted by thicker arrows:\\
	\begin{dependency}[theme=simple, arc edge , arc angle = 20]
		\begin{deptext}[column sep=0.2cm]
			\(w_b\) \& \(=\) \& \(\bar{a}\) \& \(\bar{a}\) \& \(\bar{a}\) \& \(a\) \& \(\bar{a}\) \& \(a\) \& \(\bar{a}\) \& \(a\) \& \(a\) \& \(a\) \\
			\end{deptext}
			\depedge{3}{12}{}
			\depedge[very thick]{4}{11}{}
			\depedge{5}{6}{}
			\depedge[very thick]{7}{8}{}
			\depedge[very thick]{9}{10}{}
		\end{dependency}
\qed
\end{example}

\begin{lemma}
  \((L_D, \preceq)\) is a partial order: a reflexive, transitive and anti-symmetric relation.
	\label{lem:poleq}
\end{lemma}
\begin{definition}[harmonics and rank]
Define $(h_i)_{i\in\N}$, a sequence of words of $L_D$ given by:
\begin{align*}
	h_{0} &= \varepsilon & h_{i+1} &= \bar{a}\, h_{i}\, a \quad \bar{a}\, h_{i}\, a  \enspace , \enspace \text{for }i \geq 0
\end{align*}
We call \(h_i\) the \(i\)-\emph{th order harmonic} and collectively refer to them as \emph{harmonics}. Letting \(\hat{h}_{i} = \bar{a}\, h_{i}\, a\) we obtain the following alternative definition of \((i{+}1)\)-\emph{st order harmonic}: \(h_{i+1} = \hat{h}_{i} \; \hat{h}_{i} \).
Given \(w \in L_D\), define its \emph{rank}, denoted as \(\rank(w)\), as
the greatest harmonic order embedded in \(w\), that is the greatest \(q\geq
0\) such that \(h_q \preceq w\). Note that the rank is well-defined because \(h_0 = \varepsilon\) and \( \varepsilon \preceq w\) for all \(w\in L_D\).
\end{definition}

From now on, unless stated otherwise we assume grammars to be in Chomsky
normal form. A grammar \(G = (V,\Sigma,S,\pr) \) is in \emph{Chomsky normal form} if
each production rule \(p\) of \(R\) is such that \( p = X \rightarrow Y\, Z \)
or \( p = X \rightarrow b \) where \(X,Y\) and \(Z\) are variables and \(b\) is
a terminal.

A parse tree of a grammar in Chomsky normal form has the following property: all interior nodes have one or two children where
the nodes with one child correspond to a rule of the form \(X \rightarrow b\) and every other interior node has two children such
that the three nodes correspond to a rule of the form \(X \rightarrow Y\, Z\).

Next, we give a mapping of quasi-trees onto Dyck words based on the pre-order
traversal of a tree. Given a quasi-tree \(t\), we define its \textbf{\emph{footprint}}, denoted \(\alpha(t)\), inductively as follows:
\begin{compactitem}
\item If \(n\) is a leaf then $\alpha(n)= a$.\vspace{-1em}
\item If \(n\) has \(k\) children \(n_1\) to \(n_k\) (in that order) then 
	\( \alpha(n)= a\, \overbrace{\bar{a}\, \bar{a}\, \ldots \, \bar{a}}^{k \text{ times}} \alpha(n_1)\, \alpha(n_2)\, \ldots \, \alpha(n_k) \).
\end{compactitem}
Finally, if, in addition, \(n\) is the root of \(t\) then we define
\(\alpha(t)= \bar{a}\, \alpha(n) \). Our definition was inspired by a
particular formulation of the Chomsky-Sch\"{u}tzenberger theorem
\cite{phd/de/Wich2005}. 

We need the following notation to define and prove properties of the footprint.
Given a word \(w\in \Sigma^*\), define \(w_{\ll} = (w)_{2} \ldots
(w)_{\len{w}}\) which intuitively corresponds to shifting left \(w\). For
instance the following equalities hold \(a_{\ll} = \varepsilon\), \(abc_{\ll} =
bc\) and \( w = (w)_1 \; (w)_{\ll}\) for all words \(w\).

From the definition of $\alpha(t)$ it is easy to establish the following properties:
\begin{lemma}
Let \(t\) be a quasi-tree.
\begin{compactenum}
\item For every node \(n\) of \(t\), we have \( (\alpha(n))_{\ll} \in L_D \). In particular, when \(n\) has \(k\) children \(n_1\) to \(n_k\) we have \begin{dependency}[anchor=base, baseline=0, theme=simple, arc edge , arc angle = 10]
		\begin{deptext}
			\( (\alpha(n))_{\ll}\) \& \(=\) \& \(\bar{a}\) \& \ldots \& \(\bar{a}\) \& \( (\alpha(n_1))_1 \) \& \( (\alpha(n_1))_{\ll} \) \&\(\ldots\) \& \( (\alpha(n_k))_1 \) \& \( (\alpha(n_k))_{\ll} \)\\ 
			\end{deptext}
			\depedge{3}{9}{}
			\depedge{5}{6}{}
\end{dependency}.
Following the definition of the footprint, for \(t\) rooted at \(n\) we have \( \alpha(t) = \bar{a} \; a \;(\alpha(n))_{\ll}\), hence \( \alpha(t) \in L_D\).
\item Let \( t_1 \) be a subtree of \(t\): \(\alpha(t_1) \preceq \alpha(t)\), hence \(\rank(\alpha(t_1)) \leq \rank(\alpha(t))\).
\end{compactenum}
\label{lem:alphaprops}
\end{lemma}

Since Lemma~\ref{lem:alphaprops} shows that the footprint \(\alpha(t)\) belongs to \(L_D\), we can define \textit{the rank of the footprint} of the tree. We call this rank the \textbf{\textit{oscillation}} of the tree: \(\osc(t) = \rank(\alpha(t))\).

Using harmonics we can also formulate an equivalent, alternative definition of
dimension. For space reasons, that definition is given in the appendix.

\section{Relating Dimension and Oscillation on Trees}
In this section, we establish the following relationship between the dimension and the oscillation of a tree.
\begin{theorem}\label{th:main}
Let a grammar \(G  = (V, \Sigma, S, \pr)\) be in Chomsky normal form and let \(t\) be a parse tree of \(G\). 
We have that %
\( \osc(t) -1 \leq \dim(t) \leq 2 \osc(t)\).
\end{theorem}

\begin{proof}
The proof of both inequalities is an induction on the dimension of \( t \).

First, we prove if \( \dim (t) = d \) , then \( \osc (t) \leq d + 1\).

\subparagraph*{Basis.}  Let \( \dim (t) = 0 \). Being in Chomsky normal form, the grammar \(G\) generates only one tree \(t\) such that \(\dim(t)=0\):
it consists of two nodes, the root is labelled with the start variable and the leaf with some \(b\in\Sigma\) following a rule \(S \rightarrow b\).
The footprint of \(t\) is given by 
	\begin{dependency}[anchor=base, baseline=0, theme=simple, arc edge , arc angle = 10]
		\begin{deptext}[column sep=0.2cm]
			\(\alpha(t)\) \& \(=\) \& \(\bar{a}\) \& \(a\) \& \(\bar{a}\) \& \( a \) \\ 
			\end{deptext}
			\depedge{3}{4}{}
			\depedge{5}{6}{}
		\end{dependency} %
from which we see that \( \osc(t) = 1\). Therefore, the inequality holds for the base case.

\noindent
\textbf{Induction.}  Let \( \dim (t) = d +  1\) and call \(n_{\varepsilon}\) the root node of \(t\). Since the dimension of \(t\) is \( d + 1\), the definitions of dimension and Chomsky normal form show that there is a node \(n\) of \(t\) also of dimension \( d + 1 \) that has two children \(n_1\) and \(n_2\) of dimension \( d \).
We first show the oscillation of the tree \(t_n\) rooted at \(n\) is bounded by \(d+2\).
We know that \(\osc(t_n) = \rank( \bar{a}\, \alpha(n) )\), hence 
\(\osc(t_n) = \rank( \bar{a}\, a\, \bar{a}\, \bar{a}\, (\alpha(n_1))_1\, (\alpha(n_1))_{\ll}\; (\alpha(n_2))_1\, (\alpha(n_2))_{\ll} )\).
Moreover, it follows from the induction hypothesis that \(h_{d+2} \npreceq (\alpha(n_i))_{\ll}\) for \(i=1,2\). 
Therefore, since\\
\begin{dependency}[anchor=base, baseline=0, theme=simple, arc edge , arc angle = 10] 
\begin{deptext}
			\(\alpha(n)\) \& \(=\) \& \(a\) \& \(\bar{a}\) \& \(\bar{a}\) \& \( (\alpha(n_1))_1 \) \& \( (\alpha(n_1))_{\ll} \) \& \( (\alpha(n_2))_1 \) \& \( (\alpha(n_2))_{\ll} \)\\
			\end{deptext}
			\depedge{4}{8}{}
			\depedge{5}{6}{}
\end{dependency} following Lemma~\ref{lem:alphaprops}, we find that %
\( h_{d+3} \npreceq \bar{a}\, \alpha(n) \), hence that \( \rank(\bar{a}\, \alpha(n)) \leq d+2 \) and finally that \(\osc(t_n) \leq d+2\).

\subparagraph*{Basis.} In base case, the node \(n\) is the root of \(t\) and we are done.

\noindent
\textbf{Induction.} Now let us assume that the depth of node \(n\) is \(h\). Since \(t\) is a tree there is a unique path from \(n\) to the root
of \(t\) following the parent. The parent \(m\) of \(n\) is such that \(\dim(m)=d+1\) since \(\dim(t)=d+1\) and \(dim(n)=d+1\). Moreover, since \(t\) is the parse
tree of a grammar in Chomsky normal form we have that \(m\) has two children: \(n\) and a sibling we call \(n'\). It follows from
the definition of dimension that \(\dim(n')<d+1\). Thus we find that
\begin{dependency}[anchor=base, baseline=0, theme=simple, arc edge , arc angle = 10]
		\begin{deptext}
			\(\alpha(m)\) \& \(=\) \& \(a\) \& \(\bar{a}\) \& \(\bar{a}\) \& \( (\alpha(n))_1 \) \& \( (\alpha(n))_{\ll} \) \& \( (\alpha(n'))_1 \) \& \( (\alpha(n'))_{\ll} \)\\
			\end{deptext}
			\depedge{4}{8}{}
			\depedge{5}{6}{}
\end{dependency} or with \(n\) and \(n'\) in inverted order. 
By induction hypothesis, we have that \( h_{d+3} \npreceq (\alpha(n))_{\ll}\) and \( h_{d+2} \npreceq (\alpha(n'))_{\ll}\), hence we conclude that
\( h_{d+3} \npreceq (\alpha(m))_{\ll}\) and finally that \(\rank( (\alpha(m))_{\ll} ) \leq d+2 \). Since \(m\) is at depth \(h-1\) we can apply the
induction hypothesis to conclude that \(\osc(t) \leq d+2\). The other case with \(n\) and \(n'\) inverted is treated similarly.

To complete the proof of the theorem, we prove: if \( \dim (t) = d \), then \( \osc(t) \geq \lceil \sfrac{d}{2} \rceil \). This part of the proof is done by the induction on dimension of the parse tree, using the 2-induction principle.
\subparagraph*{Basis.}  In base case we show the inequality holds for dimensions \( 0 \) and \( 1 \).
Let \( \dim (t) = 0 \). In that case the grammar \( G \), being in Chomsky normal form, generates only one possible parse tree: it consists of two nodes, the root is labelled with the start variable S and the leaf with some \(b\in\Sigma\) following a rule \(S \rightarrow b\).
This parse tree has the footprint \( \alpha(t) = \bar{a} \;a \; \bar{a}\; a \) and \( \osc(t) = 1 \), what satisfies the inequality we want to prove.
Now let \( \dim(t) = 1 \). The parse tree of dimension \(1\) that we can construct with the minimal possible number of nodes is the following: the root \(n_{\varepsilon}\) is labelled with the start variable \(S\), \(S\) has two children \(n_1\) and \(n_2\) following the rule \( S \rightarrow BC\), and \(B\) and \(C\) have one child each following the rules \( B \rightarrow b\) and \( C \rightarrow c\) for some \( b, c \in \Sigma\).
The footprint of this tree is \( \alpha(t) = \bar{a}\; \overbrace{a\; \bar{a}\; \bar{a}\; \underbrace{a\; \bar{a}\; a}_{\alpha(n_1)}\; \underbrace{a\; \bar{a}\; a}_{\alpha(n_2)}}^{\alpha(n_{\varepsilon})}\) and \( \osc(t) = 1 \geq \sfrac{1}{2} \). Since any parse tree of \(G\) of dimension 1 will have the tree with this structure as its subtree and from the fact that \(G\) is in Chomsky normal form, from Lemma~\ref{lem:alphaprops} it follows that the oscillation of the parse trees of dimension 1 will be at least 1, and therefore always greater than \( \sfrac{1}{2} \). Hence, the inequality \( \dim(t) \leq 2 \osc(t)\) holds in the base case.

\noindent
\textbf{Induction.}
Let \( \dim (t) = d + 2 \), and assume the right inequality of the theorem is true for the trees of dimension \( d \) and \( d + 1\). If the dimension of the tree is \( d + 2\), then from the definition of dimension it 
follows there is a node \( n\) in \(t\) with dimension \( d + 2 \) that has two children \( n_1 \) and \(n_2 \) of dimension \( d + 1\), and each one of those nodes has two successors with dimension \( d \) that are also siblings. Set \(n_{11}\) and \(n_{12}\) to be those successors of the node \(n_1\), and \(n_{21}\) and \(n_{22}\) the successors of the node \(n_2\). 
We thus find that: 
	\begin{dependency}[anchor=base, baseline=0, theme=simple, arc edge , arc angle = 10]
		\begin{deptext}[column sep=0.2cm]
			\(\alpha(n)\) \& \(=\) \& \( a \) \& \(\bar{a}\) \& \(\bar{a}\) \& \(a\) \& \( (\alpha(n_1))_{\ll} \) \& \(a\) \& \( (\alpha(n_2))_{\ll} \) \\ 
			\end{deptext}
			\depedge{5}{6}{}
			\depedge{4}{8}{}
	\end{dependency}. %
It also holds 
	\begin{dependency}[anchor=base, baseline=0, theme=simple, arc edge , arc angle = 10]
		\begin{deptext}[column sep=0.2cm]
			\(\bar{a}\) \& \(\bar{a}\) \& \(a\) \& \( (\alpha(n_{i1}))_{\ll} \) \& \(a\) \& \( (\alpha(n_{i2}))_{\ll} \) \& \( \preceq \) \& \((\alpha(n_i))_{\ll} \)\\ 
			\end{deptext}
			\depedge{2}{3}{}
			\depedge{1}{5}{}
	\end{dependency} %
for \(i=1,2\).
We show that the oscillation of the tree \(t_n\) rooted at the node \(n\) is bounded from below by \( \frac{d+2}{2} \). From induction hypothesis it follows that \(h_{\lceil \sfrac{d}{2} \rceil} \preceq \bar{a} \; a\; (\alpha(n_{ij}))_{\ll}\), for \( i,j \in \{1,2\} \). Thus, due to transitivity of \((L_D, \preceq\)) we find \(\bar{a} \;h_{\lceil \sfrac{d}{2} \rceil} \; a \preceq (\alpha(n_{i}))_{\ll}\), for \(i=1,2\). Hence it follows that \(h_{\lceil \sfrac{d}{2} \rceil + 1} \preceq \alpha(n)\), and for the tree \(t_n\) it holds that \(\osc(t_n) \geq \lceil \sfrac{d}{2} \rceil + 1\). Since \(t_n\) is a subtree of \(t\), it follows that \( \osc(t) \geq \osc(t_n)\) from Lemma~\ref{lem:alphaprops}, and thus \(\dim(t) \leq 2 \osc(t)\).
\end{proof}

These bounds are tight up to \( \pm 1\). For the right inequality, define \(\Pi_h\) to be the perfect binary tree of height \(h\). We have \(\dim(\Pi_h)=h\) for all \(h\). 
However, by induction on \(h\) we find that \(\osc(\Pi_{2h-1})=\osc(\Pi_{2h})=h\) for all \(h\), hence \(\osc(\Pi_{2h-1})=\osc(\Pi_{2h})=\dim(\Pi_h)\). Therefore, the upper bound is off by 1 for perfect binary trees of odd height.

\noindent
\begin{minipage}[t]{.67\textwidth}
For the lower bound we consider the following structure of trees. We define \(P_0\) as the tree consisting of a root and one child. 
The tree \(P_n\) is defined inductively as depicted on the right.
We call \(n_{1}\) and \(n_{21}\) the roots of the first and second \(P_{n-1}\) subtree, respectively.
It is easy to see by induction on \(n\) that the dimension of the tree \(P_n\) is \(n\).  
\end{minipage}%
\hfill
\begin{minipage}[t]{.3\textwidth}
	\vspace{0pt}
\begin{tikzpicture}[thick,level/.style={sibling distance=20mm/#1},level distance=7mm]
\tikzstyle{level 2}=[sibling distance=20mm]
   \node[circle,draw](root){\(n_{\varepsilon}\)}
		child {
		node[ellipse,draw,red] {\(P_{n-1}\)}
		}
		child { 
		node[circle,draw] {\(n_{2}\)}
				child {
				node[ellipse,draw,red] {\(P_{n-1}\)}
				}
				child {
				node[circle,draw] {\(n_{22}\)}
				}
	};
\end{tikzpicture}
\end{minipage}

\noindent
We show, also by induction on \(n\), that the oscillation of the tree \(P_n\) grows with dimension. After constructing the footprint of \(P_0\), we have \(\alpha(P_0) = \bar{a} \; a\; \bar{a} \;a\), hence that \(\osc(P_0) = 1\).
For \(P_n\) we have \begin{dependency}[anchor=base, baseline=0, theme=simple, arc edge , arc angle = 10]
\begin{deptext}[column sep=0.2cm]
	\(\alpha(P_n)\) \& \(=\) \& \(\bar{a}\) \& \( a \) \& \(\bar{a}\) \& \(\bar{a}\) \& \(a\) \& \((\alpha(n_1))_{\ll}\) \& \(a\) \& \(\bar{a}\) \& \(\bar{a}\) \& \(a\) \& \( (\alpha(n_{21})_{\ll}\) \& \(a\).\\ 
	\end{deptext}
	\depedge{3}{4}{}
	\depedge{5}{9}{}
	\depedge{6}{7}{}
	\depedge{11}{12}{}
	\depedge{10}{14}{}
\end{dependency} %
From there and from the inductive hypothesis, if follows that \( h_{n+1} \preceq \alpha(P_n) \) and \( h_{n+2} \npreceq \alpha(P_n) \), hence \( \dim(P_n) = \osc(P_n) - 1 = n\). 

\section{Oscillation: from Trees to Runs}

In what follows, we map a quasi-run of a PDA onto a word of \(L_D\).
Intuitively, the mapping associates, quite naturally, \(\bar{a}\) to each push action of the run
and \(a\) to each pop.  For instance, replacing the topmost stack symbol \(\gamma\) by \(\gamma\prime\), using an action \( (q,b,\gamma) \hookrightarrow (p,\gamma\prime) \) corresponds to \(a\, \bar{a}\),
pushing two symbols using an action \( (q,b,\gamma) \hookrightarrow
(p,\gamma\prime\gamma\prime\prime) \) corresponds to \(a\, \bar{a}\, \bar{a} \), etc.
 Given the LIFO policy of the stack, we see that a quasi-run is mapped
onto a word of \(L_D\), the shortest such word being \(\bar{a}\, a\).

Next, we provide the formal definition of the mapping.
We first start by observing that quasi-runs with more than one move can always be disassembled
into a first move and subsequent quasi-runs.
We need the following notation: Given two IDs \(\ID\) and \(\ID'\) such that \(\stack(\ID) = \xi\; \stack(\ID') \) holds for
some \(\xi\in\Gamma^*\) define \(\ID / \ID' = (\state(\ID), \tape(\ID), \xi ) \).
In what follows, we formalize the disassembly of quasi-runs. We assume PDA has only one state \(q\).
\begin{lemma}[Disassembly of quasi-runs]
	Let \(r = \ID_0, \ID_1, \ldots, \ID_m\) be a quasi-run with \(m>1\). Then we can \textbf{disassemble} \(r\) into its first move \(\ID_0 \vdash \ID_1\) and \(d\) quasi-runs \(r_1, \ldots, r_d\) where \(d=\len{\stack(\ID_1)}\) as follows:
\[r_1 = \ID_{p_0}/\ID_{p_1},\ldots,\ID_{p_1}/\ID_{p_1}, \ldots, r_i = \ID_{p_{i-1}}/\ID_{p_i},\ldots,\ID_{p_i}/\ID_{p_i}, \ldots, r_d =  \ID_{p_{d-1}}/\ID_{p_d},\ldots,\ID_{p_d}/\ID_{p_d},
\] where \(p_0=1\) and \(p_1,\ldots,p_d\) are defined to be the least positions such that \(\stack(\ID_{p_i}) = \stack(\ID_{p_{i-1}})_{\ll} \) for all \(i\). Necessarily, \(p_d=m\) and each quasi-run \(r_i\) starts with \( (\stack(\ID_1))_i\) as
its initial stack content.
\label{lem:disassembly}
\end{lemma}
In the following example, we show what does the above formalized disassembly look like when the quasi-run starts with a move that pushes two symbols onto the stack.
\begin{figure}[h]
\centering
\begin{minipage}[b]{.45\textwidth}
	\includegraphics[width=\textwidth]{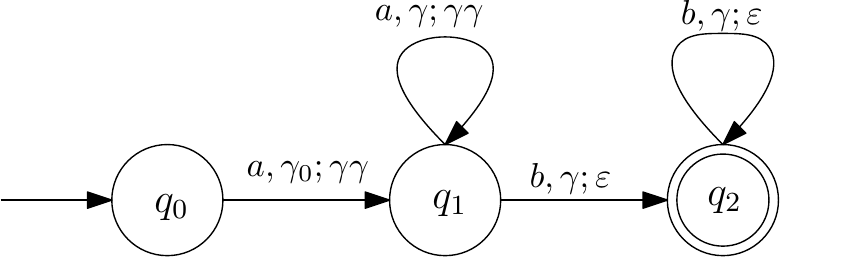}
	\caption{The automaton for Example~\ref{ex:tbc}}
	\label{fig:aut-decomp}
\end{minipage}
\begin{minipage}[b]{.45\textwidth}
	\includegraphics[width=\textwidth]{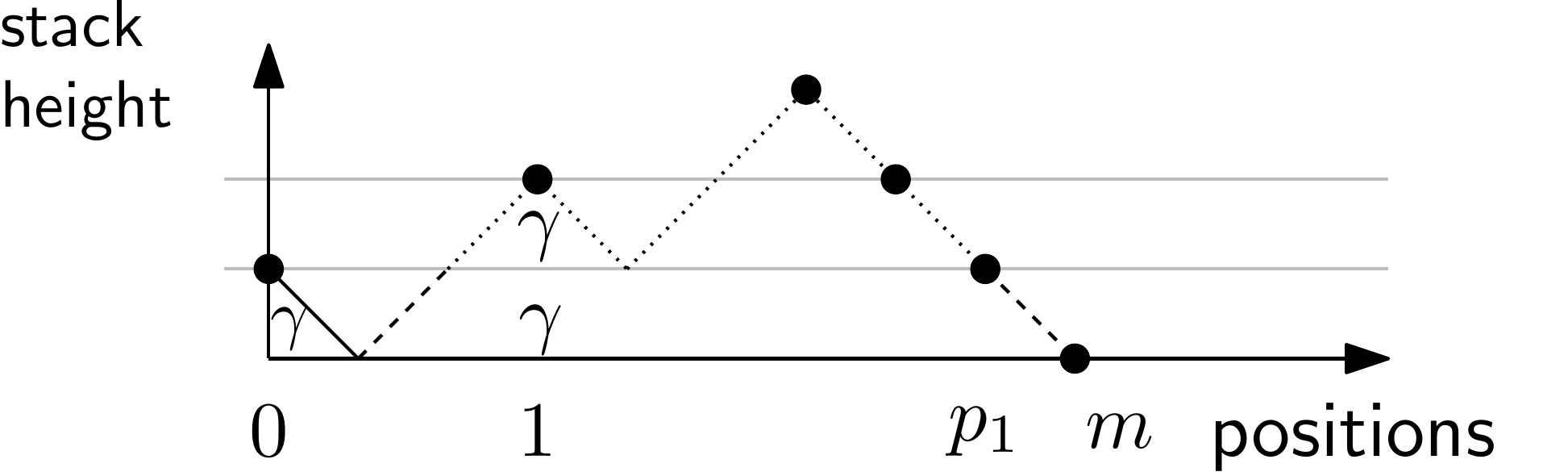}
	\caption{The quasi-run \(r\) for Example~\ref{ex:tbc}}
	\label{fig:quasi-run-decomp}
\end{minipage}
\end{figure}
\begin{example}\label{ex:tbc}
On Figure~\ref{fig:aut-decomp} we see a PDA accepting the language \( L = \{ a^n b^{n+1} \mid n \geq 1\} \). The set \(\delta\) of this PDA consists of four actions: \( (q_0, a, \gamma_0) \hookrightarrow (q_1, \gamma \gamma)\), \( (q_1, a, \gamma) \hookrightarrow (q_1, \gamma \gamma)\), \( (q_1, b, \gamma) \hookrightarrow (q_2, \varepsilon)\) and \( (q_2, b, \gamma) \hookrightarrow (q_2, \varepsilon)\), as shown on the arcs on Figure~\ref{fig:aut-decomp}. On Figure~\ref{fig:quasi-run-decomp} we see a quasi-run
\(r = \ID_0,\ldots,\ID_m\), with \(m = 5\), that accepts the word \(w = aabbb\). Each black disk
is associated with one ID in the quasi-run.
The fist move \(\ID_0 \vdash \ID_1\) yields two symbols (\(\gamma \gamma \)) on the stack. 
The first dotted line between \(1\) and \(p_1=4\) defines the quasi-run
\(r_1 = \ID_{1}/\ID_{p_1},\ldots,\ID_{p_1}/\ID_{p_1}\) while
the dashed line between \(p_1\) and \(m\) defines the quasi-run
\(r_2 = \ID_{p_1}/\ID_{m}, \ID_{m}/\ID_{m}\) that can be rewritten as
\(r_2 = \ID_{p_1}, \ID_{m}\) since \(\stack(\ID_m) = \varepsilon\).
\qed
\end{example}
Next, we define the footprint of quasi-runs based on the previous disassembly.
\begin{definition}
Given a quasi-run \(r=\ID_0,\ldots, \ID_m\) of a PDA \(P\) and its disassembly as in Lemma~\ref{lem:disassembly}, define \(\alpha'(r)\) as follows:
	\begin{compactitem}
	\item if \(m=1\) then \(\alpha'(r) = a\) \vspace{-1em}
	\item if \(m>1\) and \( \len{\stack(\ID_1)} = d \) with \(d>0\) then
		\( \alpha'(r) = a\; \overbrace{\bar{a}\ldots\bar{a}}^{d \text{ times}}\; \alpha'(r_1)\; \alpha'(r_2) \ldots \alpha'(r_d) 
		\).%
	\end{compactitem}
	Define the \emph{footprint of} \(r\), also denoted \(\alpha(r)\), as \(\alpha(r) = \bar{a}\; \alpha'(r)\).
	\label{def:quasi-run-foot}
\end{definition}
Going back to Example~\ref{ex:tbc}, definition~\ref{def:quasi-run-foot} applied
on \(r\) yields \( \alpha'(r) = a\; \bar{a}\; \bar{a}\; \alpha'(r_1)\; \alpha'(r_2) = \\ 
a\; \bar{a}\; \bar{a}\; a \; \bar{a} \; \bar{a} \; a \; a\; a\).
From now on, unless stated otherwise, to simplify the presentation, we assume the PDA \(P\) is in a \textbf{reduced form}. That is, \(P\) has only one state,
called \(q\), and each action of \(\delta\) has the following form \( (q, b,
\gamma) \hookrightarrow (q, \xi) \) where \(b \in \Sigma\cup\{\varepsilon\}\),
\(\gamma\in\Gamma\) and \(\xi \in (\Gamma^2 \cup \{\varepsilon\})\). Therefore, each action pops a symbol or pushes two symbols onto the stack.
\begin{lemma}
Let \(r\) be a quasi-run run of \(P\) in reduced form and let \(r_1, r_2\) be the disassembly of the quasi-run as in Lemma~\ref{lem:disassembly}. Then \( \osc(r) = k\) iff one of the following is satisfied:
\begin{itemize}
\item \( h_{k-1} \preceq \alpha(r_1)\) and \( \hat{h}_{k-1} \preceq \alpha(r_2)\) and \( h_k \npreceq \alpha(r_i), i=1,2\).
\item \( h_k \preceq \alpha(r_1)\) and \( h_{k+1} \npreceq \alpha(r_1)\) and \( \hat{h}_k \npreceq \alpha(r_2)\); or \( h_k \preceq \alpha(r_2)\) and \( \hat{h}_{k} \npreceq \alpha(r_2)\) and \( h_{k+1} \npreceq \alpha(r_1)\).
\end{itemize}
\label{lem:osc_disassem}
\end{lemma}
To relate footprint of trees and runs, we define a transformation from the
device generating trees (CFG) to the device generating runs (PDA).
We thus define a transformation from a grammar to pushdown automaton such that
they accept the same language (for space reason the proof is given in
appendix but the transformation is quite standard).
\begin{definition}[CFG2PDA transformation]
	Let \(G = (V, \Sigma, S, \pr)\) be a context-free grammar. Define the PDA \(P = (\{q\}, \Sigma, \Gamma, \delta, q, \gamma_0)\) where \(\Gamma=V \cup \Sigma \cup \{ \bm{e} \} \) (\(\bm{e}\notin V\cup\Sigma\)), \(\gamma_0 = S\), and, moreover, the transition function \(\delta\) consists exactly of the following actions:
\begin{compactitem}
\item \( \delta \) contains \( (q, \varepsilon, X)\hookrightarrow(q, w) \), for each rule \( (X, w) \in \pr\), with \(w \neq \varepsilon\),
\item \( \delta \) contains \( (q, \varepsilon, X)\hookrightarrow(q, \bm{e}) \), for each rule \( (X, \varepsilon) \in \pr\),
\item \( \delta \) contains \( (q, b, b)\hookrightarrow(q, \varepsilon) \), for each terminal \( b \in \Sigma\),
\item \( \delta \) contains \( (q, \varepsilon, \bm{e})\hookrightarrow(q, \varepsilon) \), for \( \bm{e} \in \Gamma\).
\end{compactitem}
\label{def:cfg2pda}
\end{definition}

\begin{proposition}
Let \(G=(V, \Sigma, S, \pr)\) be a context-free grammar and \(P=(\{q\}, \Sigma, \Gamma, \delta, q, S)\) the PDA obtained through CFG2PDA transformation. Given a parse tree \(t\), there exists a run \(r\) on input \(\yield(t)\) such that \(\alpha(t) = \alpha(r)\).
\label{prop:footprint_eq}
\end{proposition}
The proof is a simple induction on the height of \(t\). In a similar way, it can be shown that the converse also holds: starting from a run of a PDA, using the classical conversion from PDA to CFG, we obtain a CFG such that the same equality of footprints holds. This enables us to the define the oscillation of the run of a pushdown automaton in the same way we defined the oscillation of the parse tree of a context-free grammar.
\begin{definition}
	Given a \textit{(quasi-)run} \(r\) of a PDA \(P\), define its \textit{oscillation} as that
	of its footprint: \( \osc(r) = \osc(\alpha(r)) \). For \(k\geq 0\), a run
	\(r\) is said to be \(k\)-\emph{oscillating} whenever \(\osc(r) = k\).
	Define \( L^{(k)}(P)\) to be the set of words of \( L(P)\) that are accepted by some \(k\)-oscillating run. 
	We call \( L^{(k)}(P)\) the \textit{\(k\)-oscillating language of} \(P\). We say a language \(L\) is \(k\)-oscillating if there exists a PDA \(P\) such that \( L = L^{(k)}(P)\). With the term \textit{bounded-oscillation} run/language we refer to a \(k\)-oscillating run/language, when \(k\) is not important.
\label{def:OscRun}
\end{definition}

\section{Syntactic Characterization of Bounded-oscillation languages}

In this section, for a given \(k\), we define \textit{\(k\)-oscillating pushdown automaton} which we denote with \( P^{(k)} \). \( P^{(k)} \) generates (quasi-)runs of oscillation exactly \( k\). First, we give an informal description of the notation used to define \(P^{(k)}\).

The actions of \( P^{(k)} \) are derived from the actions of \(P\) by annotating the stack symbols of \(P\). 
In particular, the stack alphabet \(\Gamma'(k)\) of \(P^{(k)}\) is given by \( \Gamma'(k) := \bigcup_{i=0}^k (\Gamma^{(i)} \cup \hat{\Gamma}^{(i)})\), where \( \Gamma^{(i)} := \{ \gamma^{(i)} \mid \gamma \in \Gamma \} \) and 
\( \hat{\Gamma}^{(i)} := \{ \hat{\gamma}^{(i)} \mid \gamma \in \Gamma \} \).
Let \(\gamma' \in \Gamma'\), define \(\ann(\gamma')\) as
\( d \) if \(\gamma' \in \Gamma^{(d)} \) and \(\hat{d}\) if \(\gamma' \in \hat{\Gamma}^{(d)}\).

The goal we seek to achieve by annotating the stack alphabet is given by the next lemma.
\begin{lemma}
	Let \(r=\ID_0,\ldots,\ID_m\) be a quasi-run of \(P^{(k)}\). 
	\begin{compactitem}
	\item if \(\ann(\stack(\ID_0)) = d \) then \(r\) is \(d\)-oscillating, that is \( h_d \preceq \alpha( r )\) and \(h_{d+1} \not \preceq \alpha( r ) \)
	\item if \(\ann(\stack(\ID_0)) = \hat{d} \) then \(\hat{h}_d \preceq \alpha(r)\) and \(h_{d+1} \not
\preceq \alpha(r) \).
	\end{compactitem}
\label{lem:quasi-runs}
\end{lemma}
Next, we give the construction of \(P^{(k)}\) when \(P\) is in reduced form.  In
appendix we give the construction of \(P^{(k)}\)
given a PDA \(P\) not necessarily in reduced form.
\begin{definition}[\(k\)-oscillating pushdown automaton] 
Let \( P = (\{q\}, \Sigma, \Gamma, \delta, q, \gamma_0) \) be a PDA in reduced form, and let \( k\) be a positive integer.  We define the \textit{\(k\)-oscillating PDA} \( P^{(k)} = (\{q\}, \Sigma,
\Gamma'(k), \delta^{(k)}, q, \gamma_0^{(k)}) \) as follows: 
\begin{compactenum}
\item If \((q, b, \gamma) \hookrightarrow (q, \varepsilon)\in\delta\), then
	\( \{ (q, b, \gamma^{(0)})  \hookrightarrow (q, \varepsilon ),\; (q, b, \hat{\gamma}^{(0)}) \hookrightarrow (q, \varepsilon ) \} \subseteq \delta^{(k)} \)
\item If \((q, b, \gamma) \hookrightarrow (q, \xi) \in \delta\), with \( \xi\in \Gamma^2 \), then \(\delta^{(k)}\) contains the following actions where 
	\(d\) is any value between \(1\) and \(k\) and \(\ell\) is any value between \(0\) and \(d-1\):
\begin{compactenum}[\upshape(\itshape a\upshape)]
\item \( (q, b, \hat{\gamma}^{(d)})	\hookrightarrow (q, (\xi)_1^{(d)}\; (\xi)_2^{(\ell)}) \)
\item \( (q, b, \hat{\gamma}^{(d)}) \hookrightarrow (q, (\xi)_1^{(\ell)}\; (\hat{\xi})_2^{(d)}) \)
\item \( (q, b, \gamma^{(d)}) 			\hookrightarrow (q, (\xi)_1^{(d)}\; (\xi)_2^{(\ell)}) \)
\item \( (q, b, \gamma^{(d)}) 			\hookrightarrow (q, (\xi)_1^{(\ell)}\; (\xi)_2^{(d)}) \)
\item \( (q, b, \gamma^{(d)}) 			\hookrightarrow (q, (\xi)_1^{(d-1)}\; (\hat{\xi})_2^{(d-1)}) \)
\end{compactenum}

\end{compactenum}
\label{def_2PDA}		
\end{definition}
Let us explain the intuition behind the definition. Assume a quasi-run \(r\) of a reduced PDA and 
\begin{dependency}[anchor=base, baseline=0, theme=simple, arc edge , arc angle = 10]
\begin{deptext}[column sep=0.2cm]
	\(\alpha(r)\) \& \(=\) \& \(\bar{a}\) \& \( a \) \& \(\bar{a}\) \& \(\bar{a}\) \& \(a\) \& \((\alpha'(r_1))_{\ll}\) \& \(a\) \& \( (\alpha'(r_2))_{\ll}\).\\ 
	\end{deptext}
	\depedge{3}{4}{}
	\depedge{5}{9}{}
	\depedge{6}{7}{}
\end{dependency} %
We deduce from \(\alpha(r)\) that the first action of \(r\) is pushing two symbols yielding quasi-runs \(r_1\) and \(r_2\). We call the two symbols \((\xi)_1\) and \((\xi)_2\) as
in Definition~\ref{def_2PDA}. Now, suppose we want \(\osc(r)\geq d\), that is \(h_d \preceq \alpha(r)\).
By definition of the harmonics, we have that \(h_d = \bar{a}\; h_{d-1}\; a\ \bar{a}\; h_{d-1}\; a = \hat{h}_{d-1} \hat{h}_{d-1}\).
By Lemma~\ref{lem:osc_disassem}, one way to achieve \(\osc(r)\geq d\) is to have \( h_{d-1} \preceq
(\alpha'(r_1))_{\ll} \) and \( \hat{h}_{d-1} \preceq
(\alpha'(r_2))_{\ll}\). This situation is dealt with by the actions defined at
point 2.e.  Observe that we have a stronger requirement on
\(\alpha'(r_2)_{\ll}\) than on \(\alpha'(r_1)_{\ll}\). Indeed, we require  \(
\hat{h}_{d-1} \preceq (\alpha'(r_2))_{\ll}\) because in \(\alpha(r)\)
we have \( \bar{a}\; \alpha'(r_1)_{\ll}\; a \preceq \alpha(r) \). The purpose
of the hat annotation is to convey that stronger requirement.  This is why we
push \((\xi)_2\in\Gamma'\) such that \(\ann((\xi)_2)= \widehat{d{-}1} \).
Another way to have \(h_d \preceq \alpha(r)\) is to have \(h_d \preceq
(\alpha'(r_i))_{\ll}\) for some \(i=1,2\).  This situation is dealt with by the
actions defined at points from 2.a to 2.d.

The correctness of the transformation from \(P\) to \(P^{(k)}\) is captured by the 
next statement.

\begin{theorem}
	Let \(P\) be a PDA in reduced form and \(k\geq 0\). The following is true.
	\begin{compactenum}[\upshape(\itshape a\upshape)]
	\item If \(r\) is a run of \(P^{(k)}\) on input word \(w\) then \(\osc(r)=k\) and there is a run \(r'\) of \(P\) such that it accepts \(w\) and \( \osc(r')=k\); and  
	\item If \(r\) is a run of \(P\) on input \(w\) and \(\osc(r)=k\) then there is a run \(r'\) of \(P^{(k)}\) such that it accepts \(w\).
	\end{compactenum}
\label{th:k_runs}
\end{theorem}
Note that, as a consequence of the theorem, the following equality holds: \(L^{(k)}(P) = L(P^{(k)})\).

We conclude this section by giving upper-bounds on the size of the PDA \(
P^{(k)}\) relatively to the size of \(P\).
The size \(\len{P}\) of a PDA \(P\) is defined as \( \len{Q} + \len{\Sigma} +
\len{\Gamma} + \len{\delta} \). Relatively to \(P\), \( P^{(k)}\) has
\(\Gamma'(k)\), its stack alphabet, such that \(\len{\Gamma'(k)} = O(k \cdot
\len{\Gamma})\); \(\delta^{(k)}\), its actions, such that \(\len{\delta^{(k)}}=
O(\len{\delta}\cdot k^2) \), where \(2\) comes from the length of the sequence
\(\xi\) pushed onto the stack. Hence we find that the size of \(P^{(k)}\) is \(
O(|P|\cdot k^2) \).  The bounds for the size of \(P^{(k)}\) in the general case
and the calculations of those bounds can be found in the appendix.

\section{Operations and Decision Problems}

In this section, we first study the complexity of \(k\)-emptiness problem that
asks, for a given \(k\) and PDA \(P\), whether \(P\) has a \(k\)-oscillating run.
Then we study closure properties, for boolean operations, of the class of
bounded-oscillation languages.

\subsection{Emptiness Check}

We give a non-deterministic algorithm, called \ProcSty{query}, with three
arguments: a stack symbol \(\gamma\), an integer \(k\) and a 2-valued variable
\(h\) which can be set to HAT or NO\_HAT.  Together the three arguments
represent a stack symbol of \(P^{(k)}\), e.g., \(\gamma\),
\(k\) and HAT stand for \(\hat{\gamma}^{(k)}\).

Intuitively, the non-deterministic algorithm searches for a \(k\)-oscillating
run of \(P\) by building a run of \(P^{(k)}\) (Theorem~\ref{th:k_runs}).  It
first guesses an action \( (q,b,\gamma) \hookrightarrow (q,(\xi)_1\, (\xi)_2)\) of \(P\),
then it further guesses, through the switch statement, a case of
definition~\ref{def_2PDA}, point 2.
In a sense, given the actions of \(P\) the algorithm constructs the actions of
\(P^{(k)}\) on-the-fly.

\begin{algorithm}[h]
\ProcSty{query(}\ArgSty{\( \gamma \), $k$, $h$}\ProcSty{)}{
\\
\KwData{\(\gamma \in \Gamma\);  $k$: an integer; $h$: a 2-valued variable over \(\{\)HAT, NO\_HAT\(\}\)}
\KwResult{if \(L^{(k)}(P)\neq \emptyset\) then \ProcSty{query(}\ArgSty{\( \gamma_0 \), $k$, NO\_HAT}\ProcSty{)}
has an execution that returns; otherwise all of its executions are blocked at some \textbf{assume} statement. }
\lIf {k < 0}{\textbf{assume} false}
pick \( \xi \in \Gamma^{2} \cup \{ \varepsilon \} \)\;
\textbf{assume} \(  (q, b, \gamma) \hookrightarrow (q, \xi)\in \delta\) for some \(b\in \Sigma\cup\{\varepsilon\}\)\;
\If(\tcp*[f]{popping rule}){\( \xi = \varepsilon\)}{
\textbf{assume} \( k = 0\)\;
\KwRet\;
} 
\tcc{pushing two symbols onto the stack, \( \xi \in \Gamma^{2} \)}
	pick \( \ell \in \{0, 1, \ldots, k-1 \} \)\;
	\Switch(\tcp*[f]{non-deterministically executes one case}){*}{
		\uCase(\tcc*[f]{top stack symbol annotated with a \( \hat{} \)}){\upshape(\itshape a\upshape)}{
			\textbf{assume} \(h\) = HAT;
			\ProcSty{query(}\ArgSty{\((\xi)_2\), $\ell$, NO\_HAT}\ProcSty{)};
			\ProcSty{query(}\ArgSty{\((\xi)_1\), $k$, NO\_HAT}\ProcSty{)}\;
		}
		\uCase{\upshape(\itshape b\upshape)}{
			\textbf{assume} \(h\) = HAT;
			\ProcSty{query(}\ArgSty{\((\xi)_1\), $\ell$, NO\_HAT}\ProcSty{)};
			\ProcSty{query(}\ArgSty{\((\xi)_2\), $k$, HAT}\ProcSty{)}\;
		}
		\uCase(\tcc*[f]{top stack symbol without a \( \hat{} \) annotation}){\upshape(\itshape c\upshape)}{
			\textbf{assume} \(h\) = NO\_HAT;
			\ProcSty{query(}\ArgSty{\((\xi)_2\), $\ell$, NO\_HAT}\ProcSty{)};
			\ProcSty{query(}\ArgSty{\((\xi)_1\), $k$, NO\_HAT}\ProcSty{)}\;
		}
		\uCase{\upshape(\itshape d\upshape)}{
			\textbf{assume} \(h\){=}NO\_HAT;
			\ProcSty{query(}\ArgSty{\((\xi)_1\){,}$\ell${,}NO\_HAT}\ProcSty{)};
			\ProcSty{query(}\ArgSty{\((\xi)_2\){,}$k${,}NO\_HAT}\ProcSty{)}\;
		}
		\Case{\upshape(\itshape e\upshape)}{
			\textbf{assume} \(h\){=}NO\_HAT;
			\ProcSty{query(}\ArgSty{\((\xi)_1\){,}$k{-}1${,}NO\_HAT}\ProcSty{)};
			\ProcSty{query(}\ArgSty{\((\xi)_2\){,}$k{-}1${,}HAT}\ProcSty{)}\;
		}
	}	
}
\caption{\ProcSty{query(}\ArgSty{\( \gamma \), $k$, $h$}\ProcSty{)}, for a PDA \(P=(\{q\}, \Sigma, \Gamma, \delta, q, \gamma_0)\) in reduced form}
\label{query}
\end{algorithm}

\begin{theorem}
Given a PDA \(P\) in reduced form and a natural number \(k \geq 0\) there exists
a NSPACE(\( k \log( \len{P} )\)) decision procedure for the \(k\)-emptiness problem. %
\label{th:nspace}
\end{theorem}

\begin{proof}
We prove the following: \ProcSty{query(}\ArgSty{\( \gamma_0 \), $k$,
NO\_HAT}\ProcSty{)} returns if{}f \(L^{(k)}(P) \neq \emptyset\), or
equivalently \ProcSty{query(}\ArgSty{\( \gamma_0 \), $k$, NO\_HAT}\ProcSty{)}
has an execution that returns if{}f \(P\) has a \(k\)-oscillating run.
As usual with induction, we prove a stronger statement: 
given an ID \(\ID\) with \(\stack(\ID)\in\Gamma\),
\begin{compactitem}
\item \ProcSty{query(}\ArgSty{\( \stack(\ID) \), $k$, NO\_HAT}\ProcSty{)} returns %
if{}f %
there exists a quasi-run \(r\) from \(\ID\) such that \(h_k \preceq \alpha(r)\) and
\(h_{k+1} \npreceq \alpha(r)\); and
\item \ProcSty{query(}\ArgSty{\( \stack(\ID) \), $k$, HAT}\ProcSty{)} returns %
if{}f %
there exists a quasi-run \(r\) from \(\ID\) such that \(\hat{h}_k \preceq \alpha(r)\) and
\(h_{k+1} \npreceq \alpha(r)\).
\end{compactitem}

The proof of right-to-left direction goes by induction on the number \(m\) of
steps in the quasi-run from \(\ID\). 
If \(m=1\), then the quasi-run \(r\) is such that \(\ID \vdash \ID'\).
In this case, both  \ProcSty{query(}\ArgSty{\( \stack(\ID) \), $0$,
NO\_HAT}\ProcSty{)} and  \ProcSty{query(}\ArgSty{\( \stack(\ID) \), $0$,
HAT}\ProcSty{)} return by picking \(\xi=\varepsilon\) and the same action
as the one used to produce \(\ID'\).

Next consider \(m>1\). Since \(P\) is in reduced form the first step in \(r\) is
given by \(\ID_0 \vdash \ID_1\) where \(\stack(\ID_1)=\xi\), \(\xi\in\Gamma^2\).
Thus we can disassemble \(r\) as in Lemma~\ref{lem:disassembly}: the first move and
two quasi-runs \(r_1=J_0\ldots J_{m_j}\) and \(r_2=K_0\ldots K_{m_k}\) 
such that \(\stack(J_0)=(\xi)_1\), \(\stack(K_0)=(\xi)_2\), \(m_j, m_k < m\). First assume that \(h_k \preceq \alpha(r)\) and \(h_{k+1}\npreceq \alpha(r)\).
As before, from Lemma~\ref{lem:osc_disassem} we can reason about footprints of \(r_1\) and \(r_2\) through the two cases as listed in the lemma.
First consider the first case where \(\osc(\alpha(r_1)) = k-1\) and \(\osc(\alpha(r_2)) = k-1\) with \(\hat{h}_{k-1} \preceq \alpha(r_2)\). The induction hypothesis shows that 
\ProcSty{query(}\ArgSty{\( \stack(J_0) \), $k-1$, NO\_HAT}\ProcSty{)} and 
\ProcSty{query(}\ArgSty{\( \stack(K_0) \), $k-1$, HAT}\ProcSty{)} both return.
A close examination of \ProcSty{query} shows that
\ProcSty{query(}\ArgSty{\( \stack(\ID_0) \), $k$, NO\_HAT}\ProcSty{)} returns.

Considering the second case, the reasoning goes along the same lines as in the first one.

The proof for the case \(\hat{h}_k \preceq \alpha(r)\) and
\(h_{k+1}\npreceq \alpha(r)\) is similar.\\
For the left-to-right direction we proceed by induction on the number \(m\) of
calls to \ProcSty{query} along an execution that returns. For \(m=1\) we
necessarily have that either \ProcSty{query(}\ArgSty{\( \stack(\ID) \), $0$,
NO\_HAT}\ProcSty{)} or \ProcSty{query(}\ArgSty{\( \stack(\ID) \), $0$,
HAT}\ProcSty{)} was invoked. Then, clearly, there exists a one-move, quasi-run
\(r\) from \(\ID\) such that \(\hat{h}_0 \preceq \alpha(r)\) and
\(h_1\npreceq \alpha(r)\).

When \(m>1\), it implies that one of the case of the switch statement has been
taken followed by two calls to \ProcSty{query}. Applying the induction
hypothesis on each of these calls returns quasi-runs \(r_1\) and \(r_2\). It is
not difficult to see that they can be stitched together into a larger quasi-run
\(r\). As before, through reasoning about the footprint and relation \(\preceq\) it is possible to
show the desired property on \(\alpha(r)\).

To complete the proof of the theorem, we still have to show that the algorithm runs in NSPACE\((k\cdot\log\len{P})\). Again we observe that, in each call of query, the algorithm chooses non-deterministically between cases (a) to (e), each of which one has two calls to \ProcSty{query}. 
We observe that in all cases, the second call to \ProcSty{query} is tail-recursive and thus can be compiled away using extra variables and an unconditional jump. As for the first call, the integer being passed to \ProcSty{query} in cases (a) to (d) is \( \ell<k\).
In the case (e), the parameter that is passed to the first \ProcSty{query} call
is decreased by one: \(k-1\). Therefore, we see that along every execution we
need at most \(k\) stack frames to track the stack symbol which can be encoded
with \( \log{ \len{P} } \) bits. Hence that \( L^{(k)}(P) \neq 0\) can be
decided in NSPACE(\( k \log{ \len{P} }\)).
\end{proof}

Deciding the \(k\)-emptiness problem with \(P\) in reduced form is in NSPACE(\(
k \log( \len{P} )\)), hence it is in NLOGSPACE when \(k\) is fixed and not part
of the input. We claim that even if \(P\) is not in reduced form then
the \(k\)-emptiness problem where \(k\) is not part of the input is also in
NLOGSPACE.  To see this, we use the fact that if a decision problem \(B\) is
logspace reducible to a decision problem \(C\), and \(C \in \text{NLOGSPACE}\),
then \(B \in \text{NLOGSPACE}\) (see the proof of Lemma 4.17, point 2., in
Arora and Barak \cite{AroraB09} for LOGSPACE, the proof for NLOGSPACE is the
same). Our claim then follows from Theorem~\ref{th:nspace} and the fact we can
reduce, in deterministic logarithmic space, the \(k\)-emptiness problem for the
general form PDA to the \(k\)-emptiness problem for the reduced PDA.

Let us briefly describe how to obtain the reduced form PDA from the general
form and argue it can be computed in deterministic logarithmic space.
To obtain the reduced form PDA \(P_r\) from the general form PDA \(P\) we apply
two transformations. 

First, we reduce \(P\) to a form (denoted with \(P'\)) such that each action
either pops a symbol or pushes two symbols onto the stack.  This is done by
splitting actions whose stack words on the right-hand side have more than \(2\)
symbols into actions with stack words of exactly two symbols (or adding a
“dummy” symbol that will be pushed and immediately popped in case of a stack
word of length \(1\)). For example, the action of \(P\) \( (p, b, \gamma)
\hookrightarrow (q, \xi_1 \xi_2 \xi_3)\) yields, in \(P'\), two actions in
\(P'\) \( (p, \varepsilon, \gamma) \hookrightarrow (p_{1}, \xi'_1 \xi_3)\) and
\( (p_{1}, b, \xi'_1) \hookrightarrow (q, \xi_1 \xi_2)\), where \(p_1\) and
\(\xi'_1\) are fresh state and stack symbol.  For this, the Turing machine
enumerates the actions of \(P\), one by one, and split them when needed.
Since, to split an action, it is enough to maintain indices pointing at the
input tape, \(P'\) can be computed from \(P\) in deterministic logarithmic
space.

Second, we transform \(P'\) into \(P_r\) by encoding the states of \(P'\) into
stack symbols of \(P_r\).  From \(P'\), we thus obtain a reduced PDA \(P_r\)
with only one state, which we call \(q_r\).  Applying this transformation on
the actions given above, we obtain: \( (q_r, \varepsilon, [p \gamma r])
\hookrightarrow (q_r, [p_{1} \xi'_1 s] [s \xi_3 r])\) and \( (q_r, b, [p_{1}
\xi'_1 r]) \hookrightarrow (q_r, [q \xi_1 s] [s \xi_2 r])\) for all states \(r,
s\) of \(P'\). A Turing machine computes \(P_r\) given \(P'\)
essentially by enumerating, for each action of \(P'\), the states of \(P'\). Again
it is enough to maintain indices and thus \(P_r\) can be computed from \(P'\)
in deterministic logarithmic space.
From Lemma 4.17, point 1., in Arora and Barak \cite{AroraB09}, we conclude that
the reduction that composes the previous two can be performed in deterministic logarithmic space.

Note that reducing \(P\) to \(P_r\) clearly results in change of footprints of a run of \(P\) and the corresponding run of \(P_r\).
However, it can easily be
seen that this change of footprints will have no effect on the oscillation of the run. That is because the reduction of \(P\) affects the change of stack during a run in such a way that it only splits longer stack words into words of length \(2\) that are consecutively pushed onto the stack, so the net result stays the same.
To illustrate this, let us take a look at a simple example.
\begin{example}
Applying the action \( (p, b, \gamma)
\hookrightarrow (q, \xi_1 \xi_2 \xi_3)\) of \(P\) results in the footprint \( \alpha_1 = a \; \bar{a} \; \bar{a} \; \bar{a} \).
After that action has been split into two actions (as it has been done above), the resulting footprint is:
\( \alpha_2 = a \; \bar{a} \; \bar{a} \; a \; \bar{a} \; \bar{a} \). Observe that the symbols \((\alpha_2)_3\) and \((\alpha_2)_4\) form a matching pair (since they correspond to pushing and popping the symbol \( \xi'_1\)) and their deletion would yield the footprint equal to \( \alpha_1\). Hence that, when determining the rank of the footprint, they do not contribute to it.
\end{example}

Observe that we can modify \ProcSty{query} to solve the \emph{\(k\)-membership
problem} that asks given a word \(w\), a PDA \(P\) and a number \(k\) whether
\(w\in L^{(k)}(P)\).  The modification consists in adding an array containing
the input word \(w\), and two indices. The details are easy to recover.
\begin{lemma}
Assume a PDA \(P\) such that the maximum stack height in any run of \(P\) is \(k\). Then \(P\) is at most \(k\)-oscillating.
\label{lem:osc-and-height}
\end{lemma}
\begin{proof}
Assume \(r\) is \((k+1)\)-oscillating. Then \(h_{k+1} \preceq \alpha(r)\). Since \(h_{k+1} = \underbrace{\bar{a} \; \bar{a} \ldots \bar{a}}_{k+1} \; a \ldots \;\), the maximum stack height during \(r\) is at least \(k+1\), hence we have a contradiction.
\end{proof}

\begin{theorem}
Given a PDA \(P\) and a positive integer \(k\), the problem of deciding whether
\(P\) has a \(k\)-oscillating run is NLOGSPACE-complete when \(k\) is not part
of the input.
\label{nl-compelete}
\end{theorem}
\begin{proof}
We reduce PATH to \(k\)-emptiness problem to show it is NLOGSPACE-hard. The
problem PATH is defined as: Given a directed graph \(G\) and nodes \(s\) and
\(t\) of \(G\), is \(t\) reachable from \(s\)? PATH can be easily reduced to
the emptiness problem of a PDA \(P\) with one state \(q\). We encode the
existence of an edge between two nodes of \(G\) in actions of \(P\). For
example, if there exists an edge between the nodes \(v_1\) and \(v_2\) of
\(G\), then \(P\) has an action \( (q, \varepsilon, v_1) \hookrightarrow (q,
v_2) \). Additionally, \(P\) has an action \( (q, \varepsilon, t)
\hookrightarrow (q, \varepsilon) \). Hence, \(P\) accepts by empty stack
if{}f there is a path from \(s\) to \(t\) in \(G\). It is easy to see that the
maximum stack height during any run of \(P\) is \(1\), hence
Lemma~\ref{lem:osc-and-height} shows that \(P\) is \(1\)-oscillating. 
Finally, we conclude from Theorem~\ref{th:nspace} that the \(k\)-emptiness
problem, for a fixed \(k\) not part of the input, is in NLOGSPACE and we are
done.
\end{proof}

\subsection{Boolean Operations, Determinization and Boundedness}

We show that \(k\)-oscillating languages are closed under union, but they are not closed under intersection and complement. Also, we show that the set of deterministic CFLs (DCFLs) is not a subset of \(k\)-oscillating languages, and vice versa.
\begin{compactdesc}
\item[Union.] Let \(L^{(k)}(P_1)\) and \(L^{(k)}(P_2)\) be two
	\(k\)-oscillating languages for the PDAs \(P_1 \text{ and } P_2\),
	respectively. Then \( L^{(k)}(P_1) \cup L^{(k)}(P_2)\) is also a
	\(k\)-oscillating language for the PDA that accepts the union \(L(P_1)\cup
	L(P_2)\).
\item[Intersection.] Consider \(L_1 = \{ a^n b^n c ^j \mid n, j \geq 0\} \) and
	\(L_2 = \{ a^j b^n c^n \mid n, j \geq 0\} \). It is possible to construct
	1-oscillating PDA \(P_i\) such that \(L(P_i)=L_i\) for \(i=1,2\). However,
	\(L_1 \cap L_2 = \{ a^n b^n c^n \mid  n \geq 0 \} \) is known not to be a
	CFL.
\item[Complement.] Let \(L_1\) and \(L_2\) to be \(k\)-oscillating CFLs.
	Suppose they are closed under complement. We know \( L_1 \cap L_2 =
	\stcomp{\stcomp{L_1} \cup \stcomp{L_2}}\) holds. However, that would mean
	\(L_1\) and \(L_2\) are closed under intersection, a contradiction.
\item[Determinism.] We give an example of a DCFL which is a \(k\)-oscillating
	language for no \(k\), and vice versa. The language of even-length
	palindromes given by the grammar \(G= (\{S\}, \{0,1\}, S, \{ S \to 0S0 \mid
	1S1 \mid \varepsilon \}) \) is not a DCFL, but there exists \(1\)-oscillating
	PDA accepting it. The Dyck language \(L_D = L(G_D)\) that includes all of the harmonics \(
	\{h_i\}_{i\in\N}\) is a \(k\)-oscillating CFL for no \(k\), since harmonics
	form an infinite sequence of Dyck words. However, \(L_D\) is a DCFL. 
\item[Boundedness.] %
Given a context-free language, is it a bounded-oscillation language? 
We only sketch the proof arguments showing this problem is undecidable.
Let us start with the result of J.~Gruska who proved undecidability of the question asking whether a given context-free language is a bounded-\emph{index} language \cite{DBLP:journals/iandc/Gruska71b}. 
On the other hand, Luttenberger and Schlund~\cite{journals/iandc/LuttenbergerS16} proved a result implying the index of a CFL is bounded if{}f so is its dimension.   
Because Theorem~\ref{th:main} implies that the dimension of a CFL is bounded if{}f its oscillation is, we conclude that the problem whether a given CFL is a bounded-oscillation language is undecidable. 
\end{compactdesc}

\bibliographystyle{eptcs}

\appendix

\section{Appendix}
\subsection{Proof of Prop~\ref{prop:unambiguousDyck}}

\begin{proof}
We show that given a word of \(L(G_D)\), it has a unique parse tree or equivalently, a unique leftmost derivation.

In the proof, we denote by \(r_1\) the production rule \(S \rightarrow
\bar{a}\;S\;a\;S\) and by \(r_2\) the production \( S \rightarrow
\varepsilon\). Also, denote a step sequence with \(i\) steps by \(
\Rightarrow^i_G\). Let us assume, by contradiction, that the grammar \(G_D\) is ambiguous. Hence, there exist
two distinct leftmost derivations \(D_1\): \(S \Rightarrow^*_G w_1\) and \(D_2\): \(S
\Rightarrow^*_G w_2\) such that \( w_1 = w_2 \in \{a,\bar{a}\}^*\). Let
\(i\) be the least position in the step sequence such that 
\(D_1\) and \(D_2\) differ at \((i+1)\)-st step. Therefore, we have: \( S \Rightarrow^i_G u_i\),
and we apply \(r_i\) to \((u_i)_{p_i}\) such that there is no \(j < p_i\) for
which \((u_i)_{p_i} = S\). Since \(\len{R} = 2\), to obtain \(D_1 \neq D_2\) we
apply \(r_1\) to \(D_1\) and \(r_2\) to \(D_2\), or the other way around.
Assume the first case (the other one is treated similarly). Then we have:
\[D_1: S \Rightarrow^i_{G_D} u_i  \Rightarrow_{G_D} (u_i)_1 \ldots (u_i)_{p_i - 1} \;\bar{a} \;S \;a \;S \;(u_i)_{p_i + 1} \ldots (u_i)_{\len{u_i}} \Rightarrow^*_{G_D} w_1 \]
\[D_2: S \Rightarrow^i_{G_D} u_i  \Rightarrow_{G_D} (u_i)_1 \ldots (u_i)_{p_i - 1} \; \varepsilon \;(u_i)_{p_i + 1} \ldots (u_i)_{\len{u_i}} \Rightarrow^*_{G_D} w_2 \]
Observe that in \(D_1\) we have that \( (u_{i+1})_{p_i} = \bar{a} \) while
in \(D_2\) we have that \( (u_{i+1})_{p_i}= (u_i)_{p_i+1} = a\), hence \(w_1\neq w_2\).
\end{proof}
\subsection{Proof of Lemma~\ref{lem:poleq}}
\begin{proof}
	\begin{compactdesc}
	\item[Reflexive] If $w_1 \in L_D$, then $w_1 \preceq w_1$ for all $w_1 \in L_D$. Follows from the definition.
	\item[Transitive] For three words $w_1, w_2$ and $w_3$ from the Dyck language, if $w_1 \preceq w_2$ and $w_2 \preceq 
w_3$, it follows $w_1 \preceq w_3$. From $w_1 \preceq w_2$ it 
follows that $w_2$ can be reduced to $w_1$ by successive deletion of $n_1$ matching parentheses, and from $w_2 \preceq 
w_3$ it follows that $w_3$ can be reduced to $w_2$ by 
successive deletion of $n_2$ matching parentheses. Thus, it is possible to reduce $w_3$ to 
$w_1$ by successive deletion of $n_1 + n_2$ matching 
parentheses.
	\item[Antisymmetric] If $w_1 \preceq w_2$ and $w_2 \preceq w_1$, then $w_1 = w_2$. Follows straightforwardly from the definition of the order.\\
	\end{compactdesc}
\end{proof}
\subsection{Proof of Lemma~\ref{lem:alphaprops}}
\begin{proof}
\begin{compactitem}
\item
We prove \((\alpha(n))_{\ll} \in L_D\) by induction on the height of the tree \(t\). If the height is 0, that is, \(t\) consists of only one node \(n\), then from the definition of footprint it follows: \( \alpha(n) = a\), hence \( (\alpha(n))_{\ll} = \varepsilon \in L_D\) and we are done with the base case. Now consider a tree of height \(h+1\) with root \(n\). Assuming \(n\) has \(k\) children, following the definition of footprint we have  
\(
\alpha(n)_{\ll} = \overbrace{\bar{a}\ldots\bar{a}}^{k \text{ times}} \alpha(n_1) \ldots \alpha(n_k)
\). Next, because \(\alpha(n_i)= (\alpha(n_i))_1 \alpha(n_i)_{\ll}\), \( (\alpha(n_i))_1 = a\) and by induction hypothesis
\( \alpha(n_i)_{\ll}\in L_D\), we find that \(\alpha(n)_{\ll}\in L_D\).
\item From the inductive definition of the footprint, it follows \( \alpha(t_1) \preceq \alpha(t) \). The
	definition of rank and \((L_D, \preceq) \) being a partial order concludes the proof.
\end{compactitem}

\end{proof}
\subsection{Dimension of a Tree: a Dyck Word Based Approach}
Given a quasi-tree \(t\), we define its \emph{flattening}, denoted \(\beta(t)\), inductively as follows:
\begin{compactitem}
\item If \(n\) is a leaf then \(\beta(n)= \varepsilon \).
\item If \(n\) has \(k\) children \(n_1\) to \( n_k \) (in that order) then 
	\(\beta(n) = \bar{a}\, \beta(n_1) \, a \, \bar{a} \,\beta(n_2)\, a \, \ldots \, \bar{a}\, \beta(n_k) \, a  \).
\end{compactitem}
Finally, \(\beta(t)=\beta(n)\) where \(n\) is the root of \(t\).

It is easy to see that \(\beta(t) \in L_D\) for every tree \(t\) since the matching relation is inductively given by
\begin{dependency}[anchor=base, baseline=0, 	theme=simple, arc edge , arc angle = 25]
		\begin{deptext}[column sep=0.1cm]
			\(\beta(n)\) \& \(=\) \& \(\bar{a}\) \& \(\beta(n_1)\) \& \(a\) \& \(\bar{a}\) \& \(\beta(n_2)\) \& \(a\) \& \ldots \& \(\bar{a}\) \& \(\beta(n_k)\) \& \(a\) \\ 
			\end{deptext}
			\depedge{3}{5}{}
			\depedge{6}{8}{}
			\depedge{10}{12}{}
	\end{dependency}. %
Thus we can prove the following property.
\begin{lemma}
	Let \(t\) be a tree, \( \dim(t) = \rank(\beta(t))\).
	\label{lem:harmonic-dimension}
\end{lemma}
\begin{proof}
The proof is an induction on the height of the tree \(t\). If the height is
\(0\) and \(n\) is the root of \(t\), then \( \beta(n) = \beta(t) =
\varepsilon\). Since \(  h_0 \preceq \varepsilon \) and \( h_1 \npreceq
\varepsilon \), it follows \( \rank(\beta(t)) = 0 = \dim(t)\).\\ Next assume
the height of \(t\) is \(h+1\) and the root \(n\) of \(t\) has \(k\) children
\(n_1\) to \(n_k\). The induction hypothesis states: \( \dim(t_i) =
\rank(\beta(t_i))\) where \(t_i\) is the subtree of \(t\) with root \(n_i\).
If there is a unique maximum \(d_{max} = \max_{i\in \{1,\ldots,k\}}
\dim(t_i)\), then from the definition of dimension it follows \( \dim(t) =
d_{max}\). On the other hand, from the induction hypothesis it follows there is a
unique \(i\) such that \( \rank(\beta(t_i)) = d_{max}\). Hence, from the
definition of the rank we see that \(h_{d_{max}} \preceq \beta(t) \) and
\(h_{d_{max}+1} \npreceq \beta(t) \) , thus \(rank(\beta(t))=d_{max}\). If the
maximum is not unique, that is, \(d_{max}\) is the dimension of more than one
child of the node \(n\), then by the definition of dimension we have \(\dim(t)
= d_{max} + 1\). From induction hypothesis and the definition of the flattening
it follows that \(h_{d_{max}+1} \preceq \beta(t) \) and \(h_{d_{max}+2}
\npreceq \beta(t) \), hence that \(rank(\beta(t))=d_{max}+1\) and we are done
with the inductive case.
\end{proof}
\subsection{Proof of Lemma~\ref{lem:osc_disassem}}
\begin{proof}
One direction follows easily - if one of the above is satisfied, after writing out the footprint of \(r\) accordingly to Definition~\ref{def:quasi-run-foot}, we can easily establish that \( \osc(r) = k\).
The other direction we prove by induction on the length \(m\) of the quasi-run.
\subparagraph*{Basis.}
Necessarily, \(m=1\), \(r\) is already in disassembled form and by definition~\ref{def:quasi-run-foot} \( \alpha(r) = \bar{a} \; a\). Hence, \(\hat{h}_{0} \preceq \alpha(r)\) and \( h_1
\not \preceq  \alpha(r)\), and therefore, \( \osc(r) = 0\).
\subparagraph*{Induction.}
Following Definition~\ref{def:quasi-run-foot}, the footprint of the quasi-run \(r\) is such that:\\ 
\(\alpha(r) =  \bar{a} \; a\; \bar{a}\; \bar{a}\; \alpha'(r_1) \alpha'(r_2)\enspace ,\)
where \(r_1\) and
\(r_2\) are quasi-runs with
less than \(m\) moves
obtained through disassembly as in Lemma~\ref{lem:disassembly}. 
We can thus apply the induction hypothesis on them. 
Assume \( \osc(r) = k\). Then we can reason about the oscillation of the quasi-runs \(r_1\) and \(r_2\) by distinguishing the following two cases:
\begin{compactitem}
\item
In the first case, we can apply induction hypothesis to \(r_1\) and \(r_2\) and we have \( h_{k-1} \preceq \alpha(r_1)\), \( \hat{h}_{k-1} \preceq \alpha(r_2)\) and \( h_{k} \npreceq \alpha(r_i), i=1,2\). We thus find that \(r_1\) is \((k-1)\)-oscillating and \(r_2\) is \((k-1)\)-oscillating.
Going back to 
\(\alpha(r) =  \bar{a} \; a\; \bar{a}\; \bar{a}\; \alpha'(r_1) \alpha'(r_2)\),
we find that \(h_k \preceq \alpha(r)\) since \(h_{k-1} \preceq \alpha(r_1)\) and 
\(\hat{h}_{k-1} \preceq \alpha(r_2)\). We also find that \( h_{k+1} \npreceq \alpha(r) \) since \( h_{k} \npreceq \alpha(r_i), i=1,2\).
\item
In the second case, after applying induction hypothesis, we have \( h_k \preceq \alpha(r_1)\), \( h_{k+1} \npreceq \alpha(r_1)\) and \( \hat{h}_k \npreceq \alpha(r_2)\) (the other case is treated similarly).
We thus find that \(r_1\) is \(k\)-oscillating and \(r_2\) is \(\ell \)-oscillating, with \(\ell \leq k\) (note that, if \(r_2\) is \(k\)-oscillating, then there can be no extra matching pair around \(h_k\) in \(\alpha(r_2)\)).
Going back to 
\(\alpha(r) =  \bar{a} \; a\; \bar{a}\; \bar{a}\; \alpha'(r_1) \alpha'(r_2)\),
we find that \(h_k \preceq \alpha(r)\) since \( h_k \preceq \alpha(r_1) \) and \(h_{k+1} \npreceq \alpha(r)\) since \( \hat{h}_{k} \npreceq \alpha(r_2) \).
\end{compactitem}
Notice that these are the only possibilities that yield the quasi-run \(r\) such that \(\osc(r) = k\).
\end{proof}
\subsection{Proof of Proposition~\ref{prop:footprint_eq}}

\subsubsection{Correctness proof of the CFG2PDA transformation}
We use the following notation in the proofs that follow. Given \(w_1, w_2 \in \Sigma^*\) such that \(w_1\) is a prefix of
\(w_2\), we write \( w_1^{-1} w_2\) to denote the word \(w\in\Sigma^*\) such
that \(w_2 = w_1\, w\) holds. 
\begin{lemma}\label{lem:cfg2pda}
Let \(G = (V,\Sigma,S,\pr)\) be a grammar and \(t\) be a quasi-tree with root \(X \in V\cup\Sigma\cup\{\varepsilon\}\).
Let \(P\) be the PDA resulting from the CFG2PDA transformation
(definition~\ref{def:cfg2pda}). Then there exists a quasi-run \(\ID \vdash_P^{*}
\ID' \) such that \(\tape(\ID)=\yield(t)\) and \(\tape(\ID')=\varepsilon\). 
\end{lemma}
\begin{proof}
The proof is by induction on the height of \(t\).
\subparagraph*{Basis.} The height of \(t\) is \(0\) and therefore \(t\) consists of a single node
labelled with \(b\in \Sigma\) or \(\varepsilon\).
We conclude from definition~\ref{def:cfg2pda} that \( (q, b, b) \vdash_P (q, \varepsilon,
\varepsilon) \) and \( (q, \varepsilon, \bm{e}) \vdash_P (q, \varepsilon, \varepsilon) \) both of
which are quasi-runs with the desired properties.  
\subparagraph*{Induction.}
Now assume the height of \(t\) is \(h+1\) and the first layer has \(k\geq 1\) children and we denote by \(Y_1\) to \( Y_k \) their respective labels. Observe that each \(Y_i \in V\cup\Sigma \cup \{\varepsilon\}\) and some \(Y_j \in V\) since \(h\geq 1\). We denote with \(t_i\) the quasi-tree rooted at the node labelled with \(Y_i\).
By induction hypothesis, we have the following quasi-runs: \( (q, \yield(t_i), Y_i) \vdash^{*} (q, \varepsilon, \varepsilon)\), for \( 1 \leq i \leq k\). 
Furthermore, it is easy to see that we can “chain” the quasi-runs as follows:
\begin{align*}
	(q, \yield(t_k), Y_k) &\vdash^{*}_P  (q, \varepsilon, \varepsilon) &\text{by ind. hypothesis}\\
	(q, \yield(t_{k-1}) \yield(t_k), Y_{k-1} Y_k) &\vdash^{*} _P (q, \yield(t_k), Y_k) &\text{by ind. hyp. and PDA sem.}\\
 \vdots\\
 (q, \yield(t), Y_1\ldots Y_k) &\vdash^{*}_P (q, \yield(t_1)^{-1} \yield(t), Y_2 \ldots Y_k) & \yield(t)=\yield(t_1)\ldots \yield(t_k)
\end{align*}
Following definition~\ref{def:cfg2pda}, the production \( X \to Y_1 \ldots Y_k\) yields an action \( (q, \varepsilon, X) \hookrightarrow (q, Y_1 \ldots Y_k) \). Hence we have \( (q, \yield(t), X) \vdash_P (q, \yield(t), Y_1 \ldots Y_k)\). Putting everything together, we thus conclude
that \[ (q, \yield(t), Y_1 \ldots Y_k) \vdash^{*}_P (q, \varepsilon, \varepsilon), \] is a quasi-run with the desired properties and we are done with the inductive case.
\end{proof}
\begin{proposition}
Let \(G = (V,\Sigma,S,\pr)\) be a grammar and let \(P\) be the PDA resulting from the CFG2PDA transformation (definition~\ref{def:cfg2pda}): \(L(G) = L(P)\).
\end{proposition}
\begin{proof}
	For the left-to-right inclusion, let \(t\) be a parse tree of \(G\). It
	follows from Lemma~\ref{lem:cfg2pda} with \(X\) set to \(S\) that
	\(\stack(\ID)=S\), hence that \(P\) has a run on input \(\yield(t)\), and finally
	that \(L(G) \subseteq L(P)\).  
	
	The other direction \(L(P)\subseteq L(G)\) also
	holds. It follows from classical textbook material about the conversion
	between CFG and PDA.
\end{proof}

\begin{proof}[Proof of Proposition~\ref{prop:footprint_eq}]
The proof goes by induction on the height of \(t\). As usual with induction, we
prove a slightly different statement. First, the equality is given by \(
(\alpha(t))_{\ll} = (\alpha(r))_{\ll} \). The equality \(\alpha(t)=\alpha(r)\)
follows from the fact that \(\alpha(t),\alpha(r) \in L_D\). Second, we prove
the statement for quasi-trees and quasi-runs. More precisely, 
given a quasi-tree \(t\) %
there exists a quasi-run \( r = \ID \vdash_P^* \ID' \) such that %
\(\tape(\ID)=\yield(t)\), %
\(\tape(\ID')=\varepsilon\), and %
\( (\alpha(t))_{\ll} = (\alpha(r))_{\ll} \).

\subparagraph*{Basis.} Let \(t\) be a quasi-tree of height \(0\). 
Necessarily, \(t\) consists of a single node labelled by a terminal \(b\in\Sigma\) or \(\varepsilon\).
It follows from definition of the footprint of a quasi-tree that \(\alpha(t)_{\ll} = a\).
As we showed in the proof of lemma~\ref{lem:cfg2pda} there exists a quasi-run \(r\) given by
\( (q, b, b) \vdash_P (q,\varepsilon,\varepsilon)\) for the case \(b\in\Sigma\) and
\( (q,\varepsilon, \bm{e}) \vdash_P (q,\varepsilon,\varepsilon)\) for the case \(\varepsilon\).

In either case, the footprint of the quasi-run \(r\) is such that \( \alpha(r)_{\ll} = a\) and we are
done with the base case.

\subparagraph*{Induction.}
Suppose the quasi-tree \(t\) has height \(h+1\) and its root, labelled \(X\), has \(k\) children \(n_1\) to \(n_k\) labelled \(Y_1\) to \(Y_k\) with
\( Y_i \in V \cup \Sigma \cup \{\varepsilon\}\) for all \(i\). 
The definition of \(\alpha(t)\) shows that \( \alpha(t)_{\ll} = a\; \underbrace{\bar{a}\ldots\bar{a}}_{k \text{ times}}\; \alpha(n_1) \ldots \alpha(n_k)\).

We conclude from the proof of Lemma~\ref{lem:cfg2pda} that there exists a quasi-run \(r\) 
\begin{multline*}
(q, \yield(t), X) \vdash (q, \yield(t), Y_1\ldots Y_k) \vdash^{*} (q, \yield(t_1)^{-1} \yield(t), Y_2 \ldots Y_k) \vdash \\
\ldots \vdash	(q, \yield(t_{k-1}) \yield(t_k), Y_{k-1} Y_k) \vdash^{*}  (q, \yield(t_k), Y_k)
\vdash^{*}  (q, \varepsilon, \varepsilon) 
\end{multline*}
built upon the quasi-runs \( r_i = (q, \yield(t_i), Y_i) \vdash^{*} (q, \varepsilon, \varepsilon) \) for each \(i\).
The footprint of \(r\) is such that \( (\alpha(r))_{\ll} = a\; \underbrace{\bar{a}\ldots\bar{a}}_{k \text{ times}}\; \alpha'(r_1) \ldots \alpha'(r_k) \).
The induction hypothesis shows that \( \alpha(t_i)_{\ll} = \alpha'(r_i) \) for all \(i\), and since \( \alpha(t_i)_{\ll} = \alpha(n_i)\) we have that
\( \alpha(t)_{\ll} = \alpha(r)_{\ll} \) and we are done.
\end{proof}

\subsection{Bounded-oscillation PDA, the General Case}

\begin{definition}[\(k\)-oscillating pushdown automaton] Let \(P\) be a pushdown automaton given by \( P = (Q, \Sigma, \Gamma, \delta, 
q_0, \gamma_0) \), and let \( k\) be a fixed natural number.
We define the \textit{\(k\)-oscillating PDA} \\ \( P^{(k)} := (Q, \Sigma, \Gamma'(k), \delta^{(k)}, q_0, \gamma_0^{(k)}) \) 
and  \(\delta^{(k)}\) consists exactly of the following actions  (we assume \( b \in \Sigma \text{ or } b =  \varepsilon \)).

\begin{compactenum}
\item If \( \delta \text{ contains } (q, b, \gamma) \hookrightarrow (p, \varepsilon) \), then:

\begin{compactitem}
\item \( \delta^{(k)} \text{ contains } (q, b, \gamma^{(0)})  \hookrightarrow (p, \varepsilon ) \), and \((q, b, \hat{\gamma}^{(0)}) \hookrightarrow (p, \varepsilon ) \)
\end{compactitem}

\item If \( \delta \text{ contains } (q, b, \gamma) \hookrightarrow (p, \xi) \), with \( \xi \in \Gamma\), then
for all \( 0 < d \leq k \) we have
\begin{compactitem}
\item \( \delta^{(k)} \text{ contains } (q, b, \gamma^{(d)}) \hookrightarrow (p, \xi^{(d)} ) \), and
\((q, b, \hat{ \gamma}^{(d)}) \hookrightarrow (p, \hat{ \xi}^{(d)} ) \) 
\end{compactitem}

\item If \( \delta \text{ contains } (q, b, \gamma) \hookrightarrow (p, \xi_1 \xi_2 \ldots \xi_n) \), with \( \xi_1, \xi_2, \ldots, \xi_n \in \Gamma \) and \( n > 1 \) then \(\delta^{(k)}\) is such that it
contains \( (q, b, \nu) \hookrightarrow (p, \beta_1 \; \beta_2 \; \ldots \; \beta_n  ) \) if and only if one of the following holds:
\begin{compactenum}[\upshape(\itshape a\upshape)]
\item \( \nu = \gamma^{(d)} \) for some \( 0
	< d \leq k\) and there exists \( I \subseteq \{1, 2, \ldots, n\} \) with \( \vert I \vert
	\geq 2 \) such that \( \beta_i \in \{ \xi^{(d-1)}_i, \hat{\xi}^{(d-1)}_i \} \) for all \( i \in I\)
	and \( \beta_j \in \{ \xi^{(0)}_j, \ldots,  \xi^{(d-2)}_j
	\}\) for all \( j \notin I \). Additionally, for exactly one position \(i\) in \(I\) it holds that \(
	\beta_i = \hat{\xi}^{(d-1)}_i \) and this position can not be \(\min(I)\).
\item \( \nu = \hat{\gamma}^{(d)} \) for some \( 0
	< d \leq k\). This case is the same as above except that the set \(I\) cannot include \(\{n\}\). 
\item \( \nu =  \gamma^{(d)} \text{ or } \nu = \hat{\gamma}^{(d)} \) for some \( 0 < d \leq k\)
	and \( \beta_i \in \{ \xi^{(d)}_i, \hat{\xi}^{(d)}_i \} \) for exactly one \( i
	\in \{ 1, 2, \ldots, n \} \), and \( \beta_j \in \{
	\xi^{(0)}_j, \ldots,  \xi^{(d-1)}_j \}\) elsewhere.
	Additionally, \(\beta_i = \hat{\xi}^{(d)}_i \) if{}f \(i=n\) and \(\nu = \hat{\gamma}^{(d)}\).
\end{compactenum}

\end{compactenum}
\label{def_kPDA}		
\end{definition}
\subsubsection{The Size of \texorpdfstring{\(P^{(k)}\)}{a k-oscillating Pushdown Automaton}}
From definition~\ref{def_kPDA}, we can calculate the size of the automaton \(P^{(k)} \) constructed from the automaton \( P\). If the original automaton \(P\) has an action as at point 1 in definition~\ref{def_kPDA}, then \( P^{(k)}\) has two actions : one for the stack symbol \(\gamma\) such that \(\ann{(\gamma)}=0\), and another one in case when \(\ann{(\gamma)}=\hat{0}\).

Similarly, by looking at the point \(2\) in definition~\ref{def_kPDA} we can conclude that the PDA \(  P^{(k)}\) has \(2 k\) actions for a single action of that type in the original automaton \(P\).

The third case in the definition~\ref{def_kPDA} contributes mostly to the expansion of the size of \( P^{(k)}\). From the point 3.a in the definition~\ref{def_kPDA}, an action of \(P\) of this type gives, for a fixed \(d, 1 < d \leq k\), the following number of actions of \( P^{(k)}\):
\[ \sum_{l=2}^{n} {n \choose l} (l-1) (d-1)^{n-l}, \]
where \(n = \len{\xi}\).
For \(d = 1\), the number of actions is \(n-1\), hence the total number of actions in \( P^{(k)}\) obtained by the construction rule 3.a is:
\[ \sum_{d=2}^{k} \sum_{l=2}^{n} {n \choose l} (l-1) (d-1)^{n-l}  + (n-1).\]
In the similar fashion we calculate the number of actions of \(P^{(k)}\) for the point 3.b in definition~\ref{def_kPDA}:
\[ \sum_{d=2}^{k} \sum_{l=2}^{n-1} {n \choose l} (l-1) (d-1)^{n-l}  + (n-2).\]
From the point 3.c, we obtain that \( P^{(k)}\) has \( n d^{n-1} \) actions for a fixed \( d, 1 \leq d \leq k\) and a symbol \( \gamma\) in left side of the action such that \(\ann{(\gamma)}=d\). Similarly, it has \( n d^{n-1} \) for \(\gamma\) such that \(\ann{(\gamma)}=\hat{d}\). Thus, the total number of actions obtained from one action as at point 3.c in \(P\) is:
\[ \sum_{d=1}^{k} 2 n  d^{n-1} = 2 n \sum_{d=1}^{k} d^{n-1}. \]
To conclude, if the original PDA has \(m_1, m_2 \text{ and } m_3\) actions of type 1, 2 and 3, accordingly to the Definition~\ref{def_kPDA}, the automaton \(P^{(k)}\) has the following number of actions:
\[ 2 m_1 + 2\;k\;m_2 + m_3 \Bigg[ 2 \sum_{d=2}^{k} \sum_{l=2}^{n} {n \choose l} (l-1) (d-1)^{n-l} + (n-2) + 2 n \sum_{d=1}^{k} d^{n-1} \Bigg] \]
The above calculated size of \( \delta^{(k)}\) can be bounded above by \(O(\len{\delta} \cdot k^n)\), where \(n\) is the maximum size of all \(\xi\) through all of the actions \( (q, b, \gamma) \hookrightarrow (p, \xi)\) of \(P\).

In addition to the actions, the stack alphabet of \( P^{(k)}\) is obtained by altering and expanding the stack alphabet \(\Gamma\)  of \(P\). For each \(\gamma \in \Gamma\), we define \(\gamma^{(d)}\) and \(\hat{\gamma}^{(d)}\), for all \(d \in \{0, \ldots, k\}\). Thus \( \len{\Gamma'^{(k)}} = 2 \len{\Gamma} (k+1)\).

\subsubsection{Proof of Lemma~\ref{lem:quasi-runs}}

\begin{proof}
The proof is an induction on \(m\), the length of the run.
\subparagraph*{Basis.}
Necessarily, \(m=1\) and \( \ann{(\stack(\ID_0))} \in \{0,\hat{0}\} \) by definition~\ref{def_2PDA}.  
Also \( \alpha(r) = \bar{a} \; a\) by definition~\ref{def:quasi-run-foot}. Hence, \( \hat{h}_0 \preceq \alpha(r)\) and \( h_1
\not \preceq  \alpha(r)\), and therefore, \( \osc(r) = 0\).
\subparagraph*{Induction.}
Following definition~\ref{def:quasi-run-foot}, the footprint of the quasi-run \(r\) is such that: 
\[\alpha(r) =  \bar{a} \; a\; \bar{a}\; \bar{a}\; \alpha'(r_1) \alpha'(r_2)\enspace ,\]
where \(r_1\) and
\(r_2\) are quasi-runs with
less than \(m\) moves obtained through dissasembly as shown in Lemma~\ref{lem:disassembly}. 
We can thus apply the induction hypothesis on them. Let us rewrite them as \(r_1 = J_0, \ldots, J_{m_j}\) and \( r_2 = K_0, \ldots, K_{m_k}\).

We distinguish two cases whether \( \ann{(\stack(\ID_0))} = d \) or \(
\ann{(\stack(\ID_0))} = \hat{d}\).
\begin{compactitem}
\item By Definition~\ref{def_2PDA}, there is an action of type either (c), (d) or (e).

If we are in the case (c) (the case (d) is treated similarly) then we have
that \(\ann(\stack(J_0)) = d\), \(\ann(\stack(K_0)) = \ell\) with \(\ell\) between \(0\)
and \(d-1\).
We thus find that \(r_1\) is \(d\)-oscillating and \(r_2\) is \(\ell\)-oscillating.
Going back to 
\[\alpha(r) =  \bar{a} \; a\; \bar{a}\; \bar{a}\; \alpha'(r_1) \alpha'(r_2)\enspace ,\]
we find that \(h_d \preceq \alpha(r)\) since \(h_d \preceq \alpha(r_1)\); we
also find that \(h_{d+1}\npreceq \alpha(r)\) because \(h_{\ell+1} \npreceq \alpha(r_2)\) and
\(\ell+1 \leq d \), hence \(h_{d} \npreceq \alpha(r_2)\).

If we are in the case (e) then we have that \(\ann(\stack(J_0))=d-1\) and
\(\ann(\stack(K_0)) = \widehat{d-1}\).
We thus find that \(r_1\) is (\(d-1\))-oscillating and \(r_2\) is such that \( \hat{h}_{d-1} \preceq \alpha(r_2)\) and \( h_d \npreceq \alpha(r_2)\).
Going back to 
\(\alpha(r) =  \bar{a} \; a\; \bar{a}\; \bar{a}\; \alpha'(r_1) \alpha'(r_2)\),
we find that \(h_d \preceq \alpha(r)\) since \(\hat{h}_{d-1} \preceq
\bar{a} \; \alpha(r_1)\; a\) and \(\hat{h}_{d-1} \preceq
\alpha(r_2)\); we also find \(h_{d+1} \npreceq \alpha(r)\) since \(h_d\preceq \alpha(r_i)\) for no \(i=1,2\).

\item We switch to the case where \(\ann(\stack(\ID_0)) = \hat{d} \).
The case (a) is similar to the case (c) by making the additional observation
that \( \hat{h}_{d} \preceq \alpha(r) \) since \(h_d \preceq \alpha(r_1)\) and
\(\alpha(r_1)\) is surrounded by a matching pair (from \( (\alpha(r))_3 = \bar{a}\) to \( (\alpha'(r_2))_1 = a\)). 

Finally, for the case (b) the induction hypothesis shows that  
\(\hat{h}_{d} \preceq \alpha(r_2) \), hence \(\hat{h}_d \preceq \alpha(r) \).
We have \( h_{d+1} \npreceq \alpha(r) \) since \(h_{\ell+1} \npreceq \alpha(r_1)\) and \(\ell+1\leq d\).
\end{compactitem}
\end{proof}

\subsection{Proof of Theorem~\ref{th:k_runs}}

\begin{compactenum}[\upshape(\itshape a\upshape)]
\item Since every run is a quasi-run, follows as a direct consequence of Lemma~\ref{lem:quasi-runs}. Also by removing from the run of \(P^{(k)}\) all the annotations from the stack symbols we obtain a run of \(P\).
\item Our proof is by induction. As typical, we prove a stronger statement:
given a quasi-run \(r\) of \(P\) from \(\ID\) with \(\stack(\ID)=\gamma\) 
if \(h_k \preceq \alpha(r)\)  and \(h_{k+1}\npreceq \alpha(r)\)
then 
there is a quasi-run of \(P^{(k)}\) from \( (\state(\ID),\tape(\ID),\gamma^{(k)})\);
and if, moreover, \(\hat{h}_k \preceq \alpha(r)\)
then 
there is a quasi-run of \(P^{(k)}\) from \\
\( (\state(\ID),\tape(\ID),\hat{\gamma}^{(k)})\).

The statement of the theorem then consequently generalizes to the runs of \(P\). The proof is an induction on the length \(m\) of the quasi-run \(r\) of \(P\).
\subparagraph*{Basis.}  Since \(m=1\), \(r\) is such that \(r = \ID \vdash \ID'\). 

Hence, we have \( r = (q, b, \stack(\ID)) \vdash_P (q, \varepsilon,
\varepsilon) \) for some \(b\in\Sigma\cup\{\varepsilon\}\). Furthermore, \(\osc(r)=0\). 
Following definition~\ref{def_2PDA}, there exists a quasi-run of \(P^{(0)}\) given by \( (q, b, \stack(\ID)^{(0)}) \vdash_{P^{(0)}} (q, \varepsilon, \varepsilon) \). 
\subparagraph*{Induction.} 
Since \(P\) is in reduced form the first step in \(r\) is
given by \(\ID \vdash \ID_1\) where \(\stack(\ID_1)=\xi\), \(\xi\in\Gamma^2\).
Thus we can disassemble \(r\) as in Lemma~\ref{lem:disassembly}: the first move and
two quasi-runs \(r_1=J_0\ldots J_{m_j}\) and \(r_2=K_0\ldots K_{m_k}\) 
such that \(\stack(J_0)=(\xi)_1\), \(\stack(K_0)=(\xi)_2\), \(m_j, m_k < m\).
Hence we can apply Lemma~\ref{lem:osc_disassem} to \(r\) and we have two cases for the harmonics
embedded in the footprints of \(r_1\) and \(r_2\) as in that lemma.
Let us first look at the second case. Then \( h_k \preceq \alpha(r_i)\), \(i\) being either \(1\) or \(2\). Assume \(i=1\), other case is treated similarly. From induction hypothesis we have that there exist quasi-runs \(r_1\) of \(P^{(k)}\) from \( (\state(J_0), \tape(J_0),
(\xi)_1^{(k)}) \) and \(r_2\) of \(P^{(\ell)}\) from \( (\state(K_0), \tape(K_0),
(\xi)_1^{(l)}) \). After applying action as defined at point 2.c in Definition~\ref{def_2PDA} for the first move, we can assemble those two quasi-runs back into a quasi-run of \(P^{(k)}\).
Now look at the first case. Similarly, from induction hypothesis applied to \(r_1\) and \(r_2\) and by action as defined at point 2.e in Definition~\ref{def_2PDA}, we obtain a quasi-run of \(P^{(k)}\).
The case where \( \ann(\stack(\ID_0)) = \hat{k}\) is treated similarly.
\end{compactenum}

\end{document}